\documentclass[10pt]{article}

\usepackage{amsmath}
\usepackage{amssymb}

\usepackage{graphicx}

\usepackage{cite}

\usepackage{color} 

\usepackage[labelfont=bf,labelsep=period,justification=raggedright]{caption}

\makeatletter
\renewcommand{\@biblabel}[1]{\quad#1.}
\makeatother

\date{}

\usepackage{amsthm}
\usepackage{thmtools}
\usepackage{thm-restate}
\usepackage{nameref}
\usepackage{float}

\newcommand{\beq}{\begin{equation}}
\newcommand{\eeq}{\end{equation}}
\newcommand{\barr}{\begin{array}}
\newcommand{\earr}{\end{array}}
\newcommand{\beqr}{\begin{eqnarray}}
\newcommand{\eeqr}{\end{eqnarray}}
\newcommand{\beqrn}{\begin{eqnarray*}}
\newcommand{\eeqrn}{\end{eqnarray*}}
\newcommand{\beqn}{\begin{equation*}}
\newcommand{\eeqn}{\end{equation*}}
\newcommand{\bei}{\begin{itemize}}
\newcommand{\beii}{\begin{itemize} \item}
\newcommand{\eei}{\end{itemize}}
\newcommand{\ben}{\begin{enumerate}}
\newcommand{\een}{\end{enumerate}}
\newcommand{\bes}{\begin{small}}
\newcommand{\ees}{\end{small}}
\newcommand{\bec}{\begin{center}}
\newcommand{\eec}{\end{center}}
\newcommand{\betab}{\begin{tabular}}
\newcommand{\eetab}{\end{tabular}}
\newtheorem{theorem}{Theorem}[section]
\newtheorem{lemma}[theorem]{Lemma}
\newtheorem{proposition}[theorem]{Proposition}

\theoremstyle{definition}

\theoremstyle{remark}
\newtheorem{remark}{Remark}

\newcommand{\EV}{\mathbf{E}} 
\newcommand{\EVb}[1]{\EV\left\{#1\right\}} 
\newcommand{\cov}{\mathrm{cov}}
\newcommand{\var}{\mathrm{var}}
\newcommand{\tr}{\mathrm{tr}}
\newcommand{\rank}{\mathrm{rank}}
\newcommand{\vspan}{\mathrm{span}}

\newcommand{\diag}{\mathrm{diag}}
\newcommand{\Flin}{\mathrm{F,lin}}
\newcommand{\mutG}{\mathrm{mut,G}}
\newcommand{\OLE}{\mathrm{OLE}}
\newcommand{\tnorm}[1]{\lVert#1\rVert_2}
\newcommand{\norm}[1]{\lVert#1\rVert}
\newcommand{\abs}[1]{\lvert#1\rvert}

\newcommand{\rhosig}{\rho^{\text{sig}}}
\newcommand\bigzero{\makebox(0,0){\text{\huge0}}}

\newcommand{\namerefquote}[1]{``\nameref{#1}"}

\makeatletter
\renewcommand{\@biblabel}[1]{\quad#1.}
\makeatother

\begin{document}

\begin{flushleft}
{\Large
\textbf{The sign rule and beyond:  Boundary effects, flexibility, and noise correlations in neural population codes}
}
\\
Yu Hu$^{1,\ast}$, 
Joel Zylberberg$^{1}$, 
Eric Shea-Brown$^{1,2,3}$
\\
\bf{1}  Department of Applied Mathematics, University of Washington, Seattle, Washington, United States of America
\\
\bf{2} Program in Neurobiology and Behavior, University of Washington, Seattle, Washington, United States of America
\\
\bf{3} Department of Physiology and Biophysics, University of Washington, Seattle, Washington, United States of America
\\
$\ast$ E-mail: huyu@uw.edu
\end{flushleft}

\section*{Abstract}
Over repeat presentations of the same stimulus, sensory neurons show variable responses. This ``noise" is typically correlated between pairs of cells, and a question with rich history in neuroscience is how these noise correlations impact the population's ability to encode the stimulus.

Here, we consider a very general setting for population coding, investigating how information varies as a function of noise correlations, with all other aspects of the problem -- neural tuning curves, etc. -- held fixed. This work yields unifying insights into the role of noise correlations. These are summarized in the form of theorems, and illustrated with numerical examples involving neurons with diverse tuning curves.  Our main contributions are as follows.

(1) We generalize previous results to prove a {\it sign rule} (SR) --- if noise correlations between pairs of neurons have opposite signs vs. their signal correlations, then coding performance will improve compared to the independent case. This holds for three different metrics of coding performance, and for arbitrary tuning curves and levels of heterogeneity. This generality is true for our other results as well.

(2)
As also pointed out in the literature, the SR does not provide a necessary condition for good coding. We show that a diverse set of correlation structures can improve coding. Many of these violate the SR, as do experimentally observed correlations. There is structure to this diversity: we prove that the optimal correlation structures must lie on boundaries of the possible set of noise correlations.

(3) We provide a novel set of necessary and sufficient conditions, under which the coding performance (in the presence of noise) will be as good as it would be if there were no noise present at all.

\section*{Author Summary}
Sensory systems communicate information to the brain --- and brain areas communicate between themselves --- via the electrical activities of their respective neurons. These activities are ``noisy:" repeat presentations of the same stimulus do not yield to identical responses every time. Furthermore, the neurons' responses are not independent: the variability in their responses is typically correlated from cell to cell.  How does this change the impact of the noise --- for better or for worse?  

Our goal here is to classify (broadly) the sorts of noise correlations that are either good or bad for enabling populations of neurons to transmit information. This is helpful as there are many possibilities for the noise correlations,  and the set of possibilities becomes large for even modestly sized neural populations. We prove mathematically that, for larger populations, there are many highly diverse ways that favorable correlations can occur. These often differ from the noise correlation structures that  are typically identified as beneficial for information transmission -- those that follow the so-called ``sign rule."  Our results help in interpreting some recent data that seems puzzling from the perspective of this rule.

\section*{Introduction}
Neural populations typically show correlated variability over repeat presentation of the same stimulus~\cite{mastronarde83,alonso96,Cohen:2011eh,gawne93}. These are called \emph{noise correlations},  to differentiate them from correlations that arise when neurons respond to similar features of a stimulus.  Such \emph{signal correlations} are measured by observing how pairs of mean (averaged over trials) neural responses co-vary as the stimulus is changed~\cite{Cohen:2011eh,Averbeck:2006ew}.

How do noise correlations affect the population's ability to encode information? This question is well-studied~\cite{Zohary:1994ei,Averbeck:2006ew,Abbott:1999ul,Ecker:2011bx, shamir06, Sompolinsky:2001hh,Averbeck:2006vj,Cohen:2011eh, Roudi_Latham,romo03,daSilveira:2013vf,wilke02,Josic:2009du}, and prior work indicates that the presence of noise correlations can either improve stimulus coding, diminish it, or have little effect (Fig.~\ref{F:tuning}). Which case occurs depends richly on details of the signal and noise correlations, as well as the specific assumptions made. For example ~\cite{shamir06, Ecker:2011bx,daSilveira:2013vf} show that a classical picture --- wherein positive noise correlations prevent information from increasing linearly with population size --- does not generalize to heterogeneously tuned populations.  Similar results were obtained by ~\cite{tkacik2010}, and these examples emphasize the need for general insights. 

Thus, we study a more general mathematical model, and investigate how coding performance changes as the noise correlation are varied. Figure~\ref{F:tuning}, modified and extended from \cite{Averbeck:2006ew}, illustrates this process.  In this figure, the only aspect of the population responses that differs from case to case are the noise correlations, resulting in differently shaped distributions. These different noise structures lead to  different levels of stimulus discriminability, and hence coding performance.  The different cases illustrate our approach:  given any set of tuning curves and noise variances, we study how encoded stimulus information varies with respect to the set of all pairwise noise correlations.

Compared to previous work in this area, there are two key differences that makes our analysis novel: we make no particular assumptions on the structure of the tuning curves; and we do not restrict ourselves to any particular correlation structure such as the ``limited-range" correlations often used in prior work~\cite{Averbeck:2006ew,Ecker:2011bx,Abbott:1999ul}. Our results still apply to the previously-studied cases, but also hold much more generally. This approach leads us to derive mathematical theorems relating encoded stimulus information to the set of pairwise noise correlations. We prove the same theorems for several common measures of coding performance: the linear Fisher information, the precision of the optimal linear estimator (OLE~\cite{Salinas:1994wr}), and the mutual information between Gaussian stimuli and responses.

First, we prove that coding performance is always enhanced -- relative to the case of independent noise -- when the noise and signal correlations have opposite signs for all cell pairs (see Fig.~\ref{F:tuning}).  This ``sign rule" (SR)  generalizes prior work. Importantly, the converse is not true, noise correlations that perfectly violate the SR --and thus have the same signs as the signal correlations -- can yield better coding performance than does independent noise.  Thus, as previously observed~\cite{Ecker:2011bx, shamir06,daSilveira:2013vf}, the SR does not provide a necessary condition for correlations to enhance coding performance. 

Since experimentally observed noise correlations often have the same signs as the signal correlations~\cite{Zohary:1994ei, Cohen:2011eh, kohn05}, new theoretical insights are needed. To that effect, we develop a new organizing principle: optimal coding will always be obtained on the boundary of the set of allowed correlation coefficients.  As we discuss, this boundary can be defined in flexible ways that incorporate constraints from statistics or biological mechanisms.

Finally, we identify conditions under which appropriately chosen noise correlations can yield coding performance as good as would be obtained with deterministic neural responses. For large populations, these conditions are satisfied with high probability, and the set of such correlation matrices is very high-dimensional. Many of them also strongly violate the SR.

\begin{figure}[H]
\begin{center}
\includegraphics[width=6in]{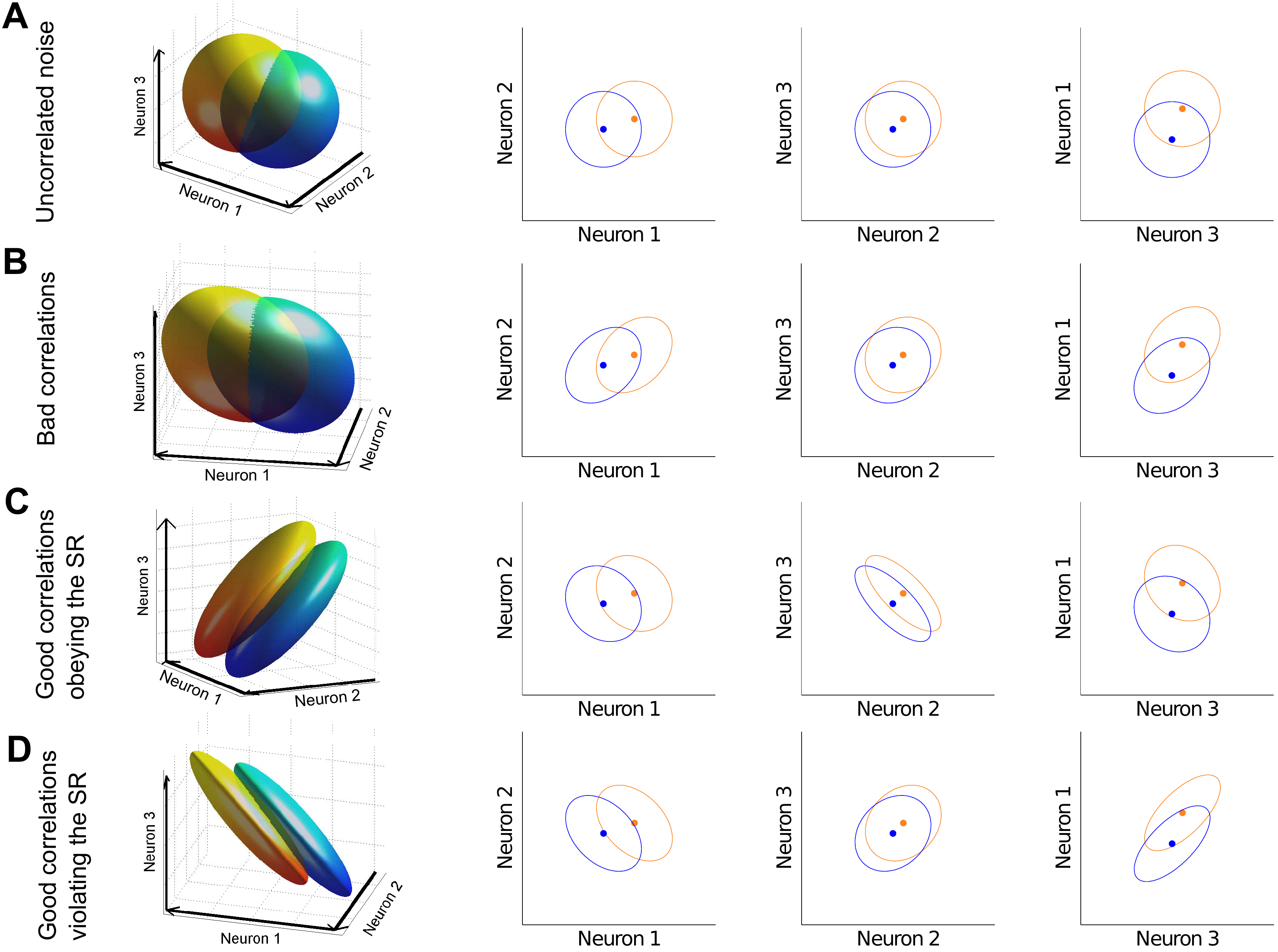}
\end{center}
\caption{\textbf{Different structures of correlated trial-to-trial variability lead to different coding accuracies in a neural population.} (Modified and extended from \cite{Averbeck:2006ew}.) 
We illustrate the underlying issues via a three neuron population, encoding two possible stimulus values (yellow and blue).  Neurons' mean responses are different for each stimulus, representing their tuning.  Trial-to-trial variability (noise) around these means is represented by the ellips(oid)s, which show 95\% confidence levels.  This noise has two aspects: for each individual neuron, its trial-to-trial variance; and at the population level, the noise correlations between pairs of neurons. We fix the former (as well as mean stimulus tuning), and ask how noise correlations impact stimulus coding.  Different choices ({\bf A-D}) of noise correlations affect the orientation and shape of response distributions in different ways, yielding different levels of overlap between the full (3D) distributions for each stimulus.  The smaller the overlap, the more discriminable are the stimuli and the higher the coding performance. We also show the 2D projections of these distributions, to facilitate the comparison with the geometrical intuition of~\cite{Averbeck:2006ew}. First, Row {\bf A} is the reference case where neurons' noise is independent:  zero noise correlations.  Row {\bf B} illustrates how noise correlations can increase overlap and worsen coding performance.  Row {\bf C} demonstrates the opposite case, where noise correlations are chosen consistently with the sign rule (SR) and information is enhanced compared to the independent noise case.  Intriguingly, Row {\bf D} demonstrates that there are more favorable possibilities for noise correlations:  here, these violate the SR, yet improve coding performance vs. the independent case.  Detailed parameter values are listed in Methods Section~\namerefquote{S:numerical_details}.}
\label{F:tuning}
\end{figure}

\section*{Results}
The layout of our Results section is as follows. We will begin by describing our setup, and the quantities we will be computing, in Section~\namerefquote{S:setup}.

 In Section~\namerefquote{S:sign_rule}, we will then discuss our generalized version of the ``sign rule", Theorem \ref{TH:p_or_n}, namely that signal and noise correlations between pairs of neurons with opposite signs will  always improve encoded information compared with the independent case. Next, in Section~\namerefquote{S:boundary}, we use the fact that all of our information quantities are convex functions of the noise correlation coefficients to conclude that the optimal noise correlation structure must lie on the boundary of the allowed set of correlation matrices, Theorem \ref{TH:boundary}.

We will further observe that there will typically be a large set of correlation matrices that all yield optimal (or near-optimal) coding performance, in a numerical example of heterogeneously tuned neural populations in Section~\namerefquote{S:heter_eg}.

We prove that these observations are general in Section~\namerefquote{S:noise_cancel} by studying the noise canceling correlations (those that yield the same high coding fidelity as would be obtained in the absence of noise). We will provide a set of necessary and sufficient conditions  for correlations to be ``noise canceling", Theorem \ref{TH:0noise_cond}, and for a system to allow for these noise canceling correlations, Theorem \ref{TH:0noise_cond_diag}. Finally, we will prove a result that suggests that, in large neural populations with randomly chosen stimulus response characteristics, these conditions are likely to be satisfied, Theorem \ref{TH:0noise_prob}.

A summary of most frequent notations we use is listed in Table~\ref{T:notation}.

\subsection*{Problem setup}
\label{S:setup}

We will consider populations of neurons that generate noisy responses $\vec{x}$ in response to a stimulus $\vec{s}$.  The responses, $\vec{x}$ -- wherein each component $x_i$ represents one cell's response -- can be considered to be continuous-valued firing rates, discrete spike counts, or binary ``words", wherein each neuron's response is a 1 (``spike") or 0 (``not spike"). The only exception is that, when we consider $I_\mutG$ (discussed below), the responses must be continuous-valued.  We consider arbitrary tuning for the neurons; $\mu_i=\EVb{x_i | \vec{s}}$. For scalar stimuli, this definition of ``tuning" corresponds to the notion of a tuning curve. In the case of more complex stimuli, it is similar to the typical notion of a receptive field. Recall that the signal correlations are determined by the co-variation of the mean responses of pairs of neurons as the stimulus is varied, and thus they are determined by the similarity in the tuning functions.

As for the structure of noise across the population, our analysis allows for the general case in which the noise covariance matrix $C^n_{ij} = \text{cov} \left(x_i, x_j  | \vec{s} \right)$ (superscript $n$ denotes ``noise") depends on the stimulus $\vec{s}$. This generality is particularly interesting given the observations of Poisson-like variability~\cite{Sof+93,Britten93} in neural systems, and that correlations can vary with stimuli~\cite{kohn05,delaRocha:2007go,Cohen:2011eh,Josic:2009du}.  We will assume that the diagonal entries of the conditional covariance matrix -- which describe each cells' variance -- will be fixed, and then ask how coding performance changes as we vary the off-diagonal entries, which describe the covariance between the cell's responses (recall that the noise correlations are the pairwise covariances, divided by the geometric mean of the relevant variances $\rho_{ij}= C^n_{ij}/ \sqrt{C^n_{ii} C^n_{jj}}$).

We quantify the coding performance with the following measures, which are defined more precisely in the Methods Section~\namerefquote{S:info_def}, below. First, we consider the linear Fisher information ($I_\Flin$, Eq.~\eqref{E:Flin_def}), which measures how easy it is to separate the response distributions that result from two similar stimuli, with a linear discriminant. This is equivalent to the quantity used by \cite{Averbeck:2006vj} and \cite{Sompolinsky:2001hh} (where Fisher information reduces to $I_\Flin$). While Fisher information is a measure of \emph{local} coding performance, we are also interested in global measures. 

We will consider two such global measures, the OLE information $I_{\OLE}$ (Eq.~\eqref{E:OLE}) and mutual information for Gaussian stimuli and responses $I_\mutG$ (Eq.~\eqref{E:mutG_def}). $I_\OLE$ quantifies how well the optimal linear estimator (OLE) can recover the stimulus from the neural responses: large  $I_{\OLE}$ corresponds to small mean squared error of OLE and vice versa. For the OLE, there is one set of read-out weights used to estimate the stimulus, and those weights do not change as the stimulus is varied. For contrast, with linear Fisher information, there is generally a different set of weights used for each (small) range of stimuli within which the discrimination is being performed. 

Consequently, in the case of $I_{\OLE}$ and $I_\mutG$, we will be considering the average noise covariance matrix $C^n_{ij} =\cov(x_i,x_j)=\EVb{\text{cov} \left(x_i, x_j  | \vec{s} \right) } $ , where the expectation is taken over the stimulus distribution. Here we overload the notation $C^n$ be the covariance matrix that one chooses during the optimization, which will be either local (conditional covariances at a particular stimulus) or global depending on the information measure we consider.

While $I_{\OLE}$ and $I_{\Flin}$ are concerned with the performance of linear decoders, the mutual information $I_{\mutG}$ between stimuli and responses describes how well the optimal read-out could recover the stimulus from the neural responses, without any assumptions about the form of that decoder. However, we emphasize that our results for $I_{\mutG}$ only apply to jointly Gaussian stimulus and response distributions, which is a less general setting than the conditionally Gaussian cases studied in many places in the literature. An important exception is that Theorem~\ref{TH:boundary} additionally applies to the case of conditionally Gaussian distributions (see discussion in Section~\namerefquote{S:convexity}).

For simplicity, we describe most results for scalar stimulus $s$ if not stated otherwise, but the theory holds for multidimensional stimuli (see Methods Section~\namerefquote{S:info_def}).

\subsection*{The sign rule revisited}
\label{S:sign_rule}

Arguments about pairs of neurons suggest that coding performance is improved -- relative to the case of independent, or trial-shuffled data -- when the noise correlations have the opposite sign from the signal correlations~\cite{Averbeck:2006ew,Abbott:1999ul,Sompolinsky:2001hh,romo03}: we dub this the ``sign rule" (SR). This notion has been explored and demonstrated in many places in the experimental and theoretical literature, and formally established for homogenous positive correlations \cite{Sompolinsky:2001hh}.  However, its applicability in general cases is not yet known.

Here, we formulate this SR property as a theorem without restrictions on homogeneity or population size. 

\begin{restatable}{theorem}{PoN}
\label{TH:p_or_n}
If, for each pair of neurons, the signal and noise correlations have opposite signs, the linear Fisher information is greater than the case of independent noise (trial-shuffled data). In the opposite situation where the signs are the same, the linear Fisher information is decreased compared to the independent case, in a regime of very weak correlations.  Similar results hold for $I_{\OLE}$ and $I_{\mutG}$, with a modified definition of signal correlations given in Section~\textnormal{\namerefquote{S:info_def}}.
\end{restatable}

\noindent In the case of Fisher information, the signal correlation between two neurons is defined as $\rhosig_{ij}=\frac{\nabla \mu_i \cdot \nabla \mu_j}{\tnorm{ \nabla \mu_i }  \tnorm{ \nabla \mu_j }}$ (Section~\namerefquote{S:info_def}). Here, the derivatives are taken with respect to the stimulus. This definition recalls the notion of the alignment in the change in the neurons' mean responses in, e.g., \cite{Averbeck:2006vj}. It is important to note that this definition for signal correlation is locally defined near a stimulus value; thus, it differs from some other notions of ``signal correlation" in the literature, that quantify how similar the whole tuning curves are for two neurons (see discussion on the alternative $\tilde{\rho}^{\text{sig}}_{ij}$ in Section~\namerefquote{S:info_def}). We choose to define signal correlations for $I_\Flin$, $I_\OLE$ and $I_\mutG$ as described in Section~\namerefquote{S:info_def} to reflect precisely the mechanism behind the examples in \cite{Averbeck:2006ew}, among others.

It is a consequence of Theorem~\ref{TH:p_or_n} that the SR holds pairwise; different pairs of neurons will have different signs of noise correlations, as long as they are consistent with their (pairwise) signal correlations. The result holds as well for heterogenous populations. The essence of our proof of Theorem \ref{TH:p_or_n} is to calculate the gradient of the information function in the space of noise correlations.  We compute this gradient at the point representing the case where the noise is independent. The gradient itself is determined by the signal correlations, and will have a positive dot product with any direction of changing noise correlations that obeys the sign rule. Thus, information is increased by following the sign rule, and the gradient points to (locally) the direction for changing noise correlations that maximally improves the information, for a given strength of correlations. A detailed proof is included in Methods Section~\namerefquote{S:SR_proof}; this includes a formula for the gradient direction (Remark~\ref{R:grad_Flin} in Section~\namerefquote{S:SR_proof}). We have proven the same result for all three of our coding metrics, and for both scalar, and multi-dimensional, stimuli.

Intriguingly, there exists an asymmetry between the result on improving information (above), and the (converse) question of what noise correlations are worst for population coding. As we will show later, the information quantities are convex functions of the noise correlation coefficients (see Fig.~\ref{F:2d_sign}).  As a consequence, performance will keep increasing as one continues to move along a ``good" direction, for example indicated by the SR. This is what one expects when climbing a parabolic landscape in which the second derivative is always nonnegative. The same convexity result indicates that the performance will not decrease monotonically along a ``bad" direction, such as the anti-SR direction. For example, if, while following the anti-SR direction, the system passed by the minimum of the information quantity, then continued increases in correlation magnitude would yield increases in the information. In fact, it is even possible for anti-SR correlations to yield better coding performance than would be achieved with independent noise.  An example of this is shown in Fig.~\ref{F:2d_sign}, where the arrow points in the direction in correlation space predicted by the SR, but performance that is better than with independent noise can also be obtained by choosing noise correlations in the opposite direction.

Thus, the result that anti-SR noise correlations harm coding is only a ``local" result -- near the point of zero correlations --  and therefore requires the assumption of weak correlations. We emphasize that this asymmetry of the SR is intrinsic to the problem, due to the underlying convexity.

One obvious limitation of Theorem \ref{TH:p_or_n} and the ``sign rule" results in general is that they only compare information in the presence of correlated noise with the baseline case of independent noise. This approach does not address the issue of finding the optimal noise correlations, nor does it provide much insight into experimental data that do not obey the SR. Does the sign rule rule describe optimal configurations? What are the properties of the global optima? How should we interpret noise correlations that do not follow the SR? We will address these questions in the following sections.

\subsection*{Optimal correlations lie on boundaries}
\label{S:boundary}

Let us begin by considering a simple example to see what could happen for the optimization problem we described in Section~\namerefquote{S:setup}, when the baseline of comparison is no longer restricted to the case of independent noise. This example is for a population of 3 neurons. In order to better visualize the results, we further require that $C^n_{1,2}=C^n_{1,3}$. Therefore the configurations of correlations is two dimensional. In Fig.~\ref{F:2d_sign}, we plot information $I_\OLE$ as a function of the two free correlation coefficients (in this example the variances are all $C^n_{ii}=1$, thus $C^n_{ij}=\rho_{ij}$). 

First, notice that there is a parabola-shaped region of all attainable correlations (in Fig.~\ref{F:2d_sign}, enclosed by black dashed lines and the upper boundary of the square). The region is determined not only by the entry-wise constraint $ \vert \rho_{i,j} \vert \leq 1$ (the square), but also by a global constraint that the covariance matrix $C^n$ must be positive semidefinite. For linear Fisher information and mutual information for Gaussian distributions, we further assume $C^n\succ 0$ (i.e. $C^n$ is positive definite) so that $I_\Flin$ and $I_\mutG$ remain finite (see also Section~\namerefquote{S:info_def}).  As we will see again below, this important constraint leads to many complex properties of the optimization problem. This constraint can be understood by noting that correlations must be chosen to be ``consistent" with each other and cannot be freely and independently chosen. For example, if $\rho_{1,2} = \rho_{1,3}$ are large and positive, then cells 2 and 3 will be positively correlated -- since they both covary positively with cell 1 -- and $\rho_{2,3}$ may thus not take negative values. In the extreme of $\rho_{1,2} = \rho_{1,3}=1$, $\rho_{2,3}$ is fully determined to be 1. Cases like this are reflecting the corner shape in the upper right of the allowed region in Fig.~\ref{F:2d_sign}. 

The case of independent noise is denoted by a black dot in the middle of Fig.~\ref{F:2d_sign}, and the gradient vector of $I_{\OLE}$ points to a quadrant that is guaranteed to increase information vs. the independent case (Theorem~\ref{TH:p_or_n}). The direction of this gradient satisfies the sign rule, as also guaranteed by Theorem~\ref{TH:p_or_n}. However, the gradient direction and the quadrant of the SR both fail to capture the globally optimal correlations, which are at upper right corner of the allowed region, and indicated by the red triangle. This is typically what happens for larger, and less symmetric populations, as we will demonstrate next.

\begin{figure}[H]
\begin{center}
\includegraphics[width=3.5in]{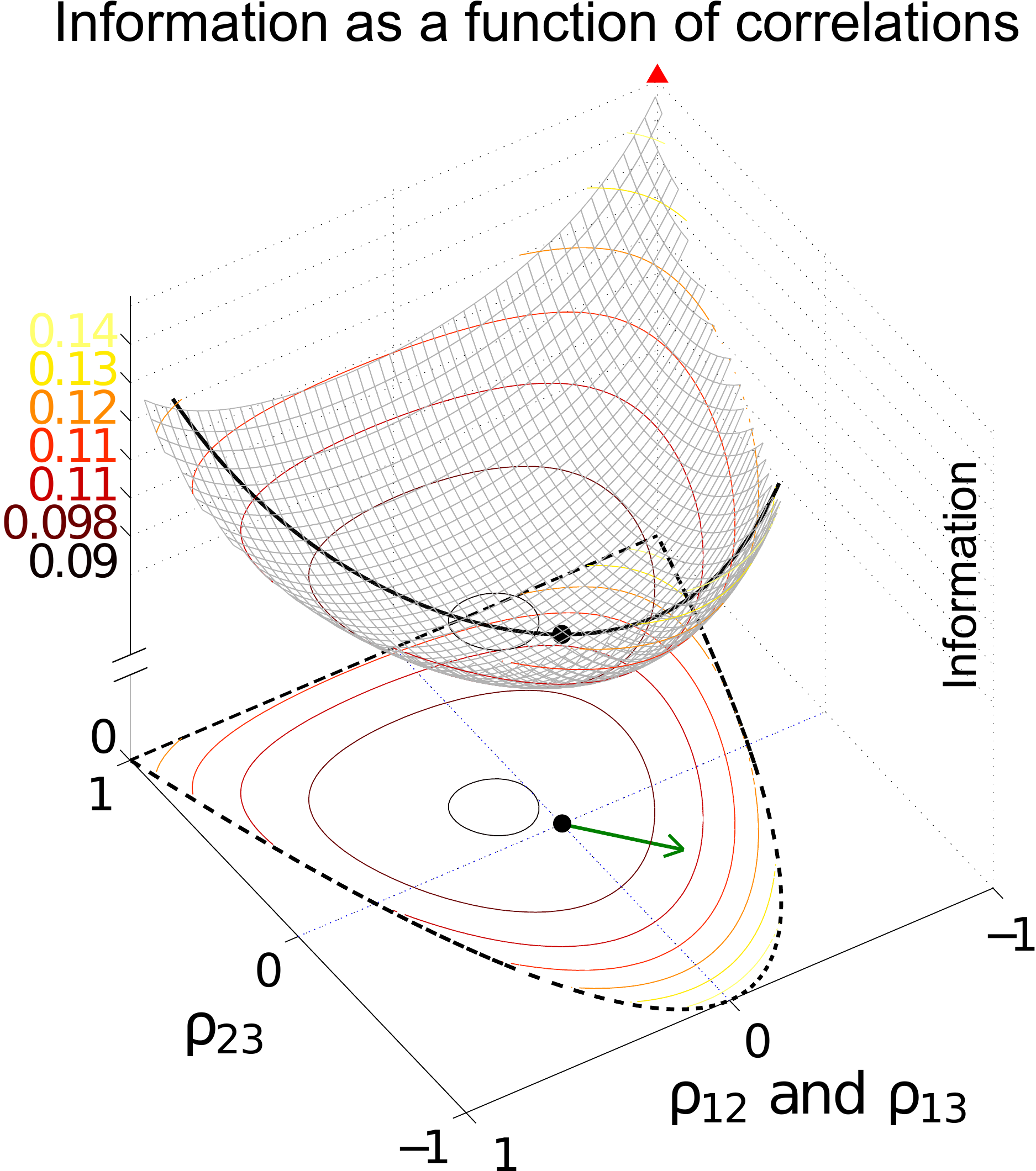}
\end{center}
\caption{ \textbf{The ``sign rule" may fail to identify the globally optimal correlations}. The optimal linear estimator (OLE) information $I_{\OLE}$ (Eq.~\eqref{E:OLE}), which is maximized when the OLE produces minimum-variance signal estimates, is shown as a function of all possible choices of noise correlations (enclosed within the dashed line).  These  values are $C^n_{1,2}=C^n_{1,3}$ (x-axis) and $C^n_{2,3}$ (y-axis) for a 3-neuron population. The bowl shape exemplifies the general fact that $I_\OLE$ is a convex function and thus must attain its maximum on the boundary (Theorem~\ref{TH:boundary}) of the allowed region of noise correlations. The independent noise case and global optimal noise correlations are labeled by a black dot and triangle respectively. The arrow shows the gradient vector of  $I_{\OLE}$, evaluated at zero noise correlations. It points to the quadrant in which noise correlations and signal correlations have opposite signs, as suggested by Theorem \ref{TH:p_or_n}.  Note that this gradient vector, derived from the ``sign rule", does not point towards the global maximum, and actually misses the entire quadrant containing that maximum. This plot is a two-dimensional slice of the cases considered in Fig.~\ref{F:3d_boundary}, while restricting $C^n_{1,2}=C^n_{1,3}$ (see Methods Section~\namerefquote{S:numerical_details} for further parameters).}
\label{F:2d_sign}
\end{figure}

Since the sign rule cannot be relied upon to indicate the global optimum, what other tools do we have at hand?  A key observation, which we prove in the Methods Section~\namerefquote{S:boundary_proof}, is that information is a convex function of the noise correlations (off-diagonal elements of $C^n$). This immediately implies:

\begin{restatable}{theorem}{Boundary}
\label{TH:boundary}
The optimal $C^n$ that maximize information must lie on the boundary of the region of correlations considered in the optimization.
\end{restatable}

As we saw in Fig.~\ref{F:2d_sign}, mathematically feasible noise correlations may not be chosen arbitrarily but are constrained by the fact that the noise covariance matrix be positive semidefinite.  We denote this condition by $C^n \succcurlyeq 0$, and recall that it is equivalent to all of its eigenvalues being non-negative. According to our problem setup, the diagonal elements of $C^n$, which are the individual neurons' response variances, are fixed. It can be shown that this diagonal constraint specifies a linear slice through the cone of all $C^n \succcurlyeq 0$, resulting a bounded convex region in $\mathbb{R}^{N(N-1)/2}$ called a \emph{spectrahedron}, for a population of $N$ neurons. These spectrahedra are the largest possible regions of noise correlation matrices that are physically realizable, and are the set over which we optimize, unless stated otherwise. 

Importantly for biological applications, Theorem \ref{TH:boundary} will continue to apply, when additional constraints define smaller allowed regions of noise correlations within the spectrahedron. These constraints may come from circuit or neuron-level factors. For example, in the case where correlations are driven by common inputs~\cite{Bin+01,delaRocha:2007go}, one could imagine a restriction on the maximal value of any individual correlation value.  In other settings, one might consider a global constraint by restricting the maximum Euclidean norm (2-norm) of the noise correlations (defined in Eq.~\eqref{E:rho_2norm} in Methods).

For a population of $N$ neurons, there are $N(N-1)/2$ possible correlations to consider;  naturally, as $N$ increases, the optimal structure of noise correlations can therefore become more complex.  Thus we illustrate the Theorem above with an example of 3 neurons encoding a scalar stimulus, in which there are 3 noise correlations to vary. In Fig.~\ref{F:3d_boundary}, we demonstrate two different cases, each with distinct $(C^\mu)_{ij}=\cov(\mu_i,\mu_j)$ matrix and vector $L_i=\cov(s,\mu_i)$ (values are given in Methods Section~\namerefquote{S:numerics}).  In the first case, there is a unique optimum (panel {\bf A}, largest information is associated with the lightest color).  In the second case, there are 4 disjoint optima (panel {\bf B}), all of which lie on the boundary of the spectrahedron. 

In the next section, we will build from this example to a more complex one including more neurons.  This will suggest further principles that govern the role of noise correlations in population coding.

\begin{figure}[H]
\begin{center}
\includegraphics[width=2.8in]{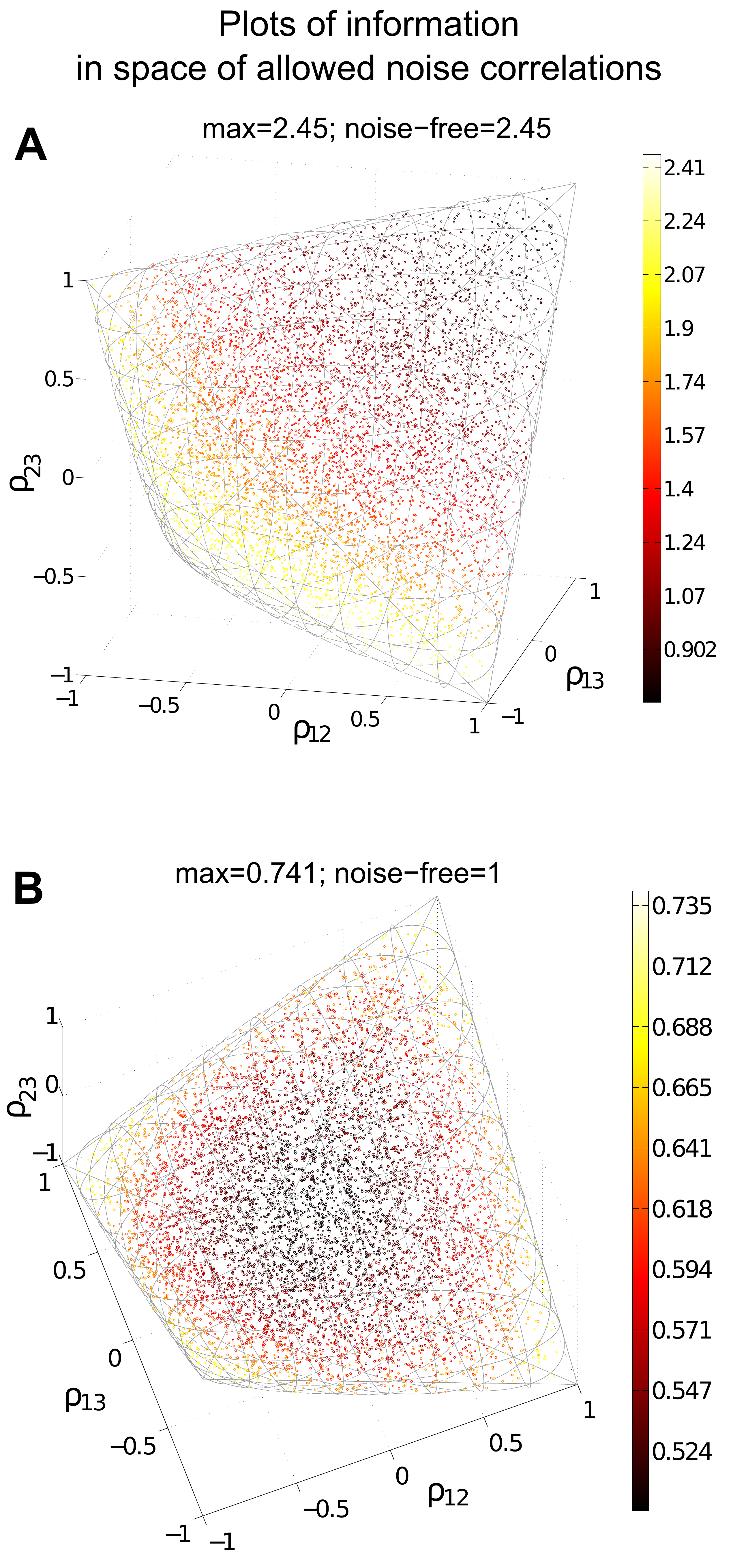}
\end{center}
\caption{\textbf{Optimal coding is obtained on the boundary of the allowed region of noise correlations.} For fixed neuronal responses variances and tuning curves, we compute coding performance -- quantified by $I_{\OLE}$ information values -- for different values of the pair-wise noise correlations. To be physically realizable, the correlation coefficients must form a positive semi-definite matrix. This constraint defines a spectrahedron, or a swelled tetrahedron, for the $N=3$ cells used. The colors of the points represent $I_{\OLE}$ information values. With different parameters $C^\mu$ and $L$ (see values in Methods Section~\namerefquote{S:numerical_details} ), the optimal configuration can appear at different locations, either unique ({\bf A}) or attained at multiple disjoint places ({\bf B}), but always on the boundary of the spectrahedron. In both panels, plot titles give the maximum value of $I_{\OLE}$ attained over the allowed space of noise correlations, and the value of $I_{\OLE}$  that would obtained with the given tuning curves, and perfectly deterministic neural responses. This provides an upper bound on the attainable $I_{\OLE}$ (see text Section~\namerefquote{S:noise_cancel}). Interestingly, in panel (\textbf{A}), the noisy population achieves this upper bound on performance, but this is not the case in (\textbf{B}). Details of parameters used are in Methods Section~\namerefquote{S:numerical_details}. }
\label{F:3d_boundary}
\end{figure}

\subsection*{Heterogeneously tuned neural populations}
\label{S:heter_eg}

We next follow~\cite{wilke02,shamir06,Ecker:2011bx} and study a numerical example of a larger ($N=20$) heterogeneously tuned neural population. The stimulus encoded is the direction of motion, which is described by a 2-D vector $\vec{s} = (\cos(\theta), \sin(\theta))^T$. We used the same parameters and functional form for the shape of tuning curves as in \cite{Ecker:2011bx}, the details of which are provided in Methods Section~\namerefquote{S:numerical_details}. The tuning curve for each neuron was allowed to have randomly chosen width and magnitude, and the trial-to-trial variability was assumed to be Poisson: the variance is equal to the mean. As shown in Fig.~\ref{F:diff_tuning}{\bf A}, under our choice of parameters the neural tuning curves -- and by extension, their responses to the stimuli -- are highly heterogenous. Once again, we quantify coding by $I_\OLE$ (see definition in Section~\namerefquote{S:setup} or Eq.~\eqref{E:OLE} in Methods).  


Our goal with this example is to illustrate two distinct regimes, with different properties of the noise correlations that lead to optimal coding.  In the first regime, which occurs closest to the case of independent noise, the SR determines the optimal correlation structure.  In the second, moving further away from the independent case, the optimal correlations may disobey the SR.  (A related effect was found by~\cite{Ecker:2011bx}; we return to this in the Discussion.)  We accomplish this in a very direct way:  for gradually increasing the (additional) constraint on the Euclidean norm of correlations (Eq.~\eqref{E:rho_2norm} in Methods Section~\namerefquote{S:info_def}), we numerically search for optimal noise correlation matrices and compare them to predictions from the SR.    

In Fig.~\ref{F:diff_tuning}{\bf B} we show the results, comparing the information attained with noise correlations that obey the sign rule with those that are optimized, for a variety of different noise correlation strengths. As they must be, the optimized correlations always produce information values as high as, or higher than, the values obtained with the sign rule.

\begin{figure}[H]
\begin{center}
\includegraphics[width=6in]{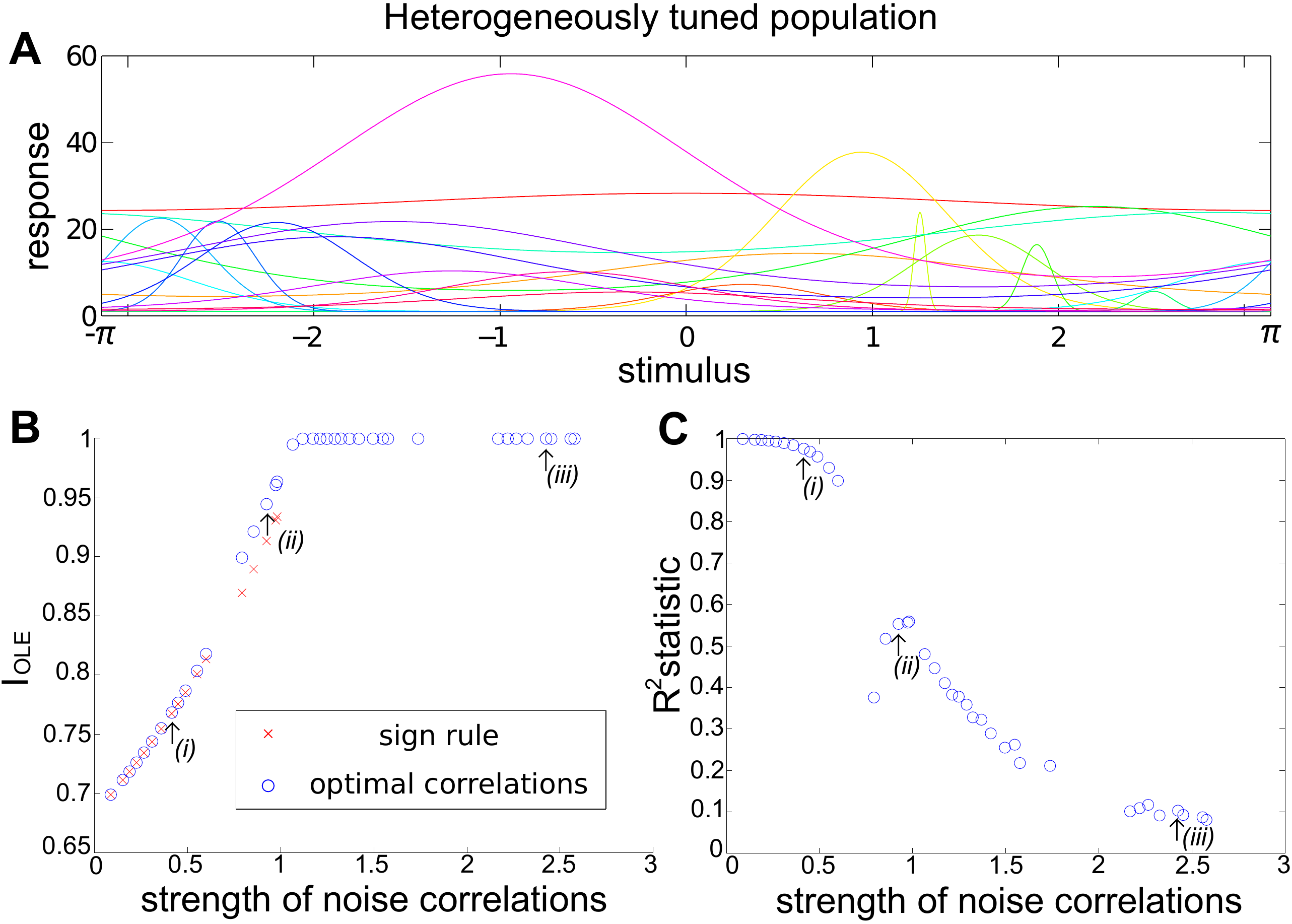}
\end{center}
\caption{\textbf{Heterogeneous neural population and violations of the sign rule with increasing correlation strength.}  We consider signal encoding in a population of 20 neurons, each of which has a different dependence of its mean response on the stimulus (heterogeneous tuning curves shown in {\bf A}).  We optimize the coding performance of this population with respect to the noise correlations, under several different constraints on the magnitude of the allowed noise correlations. Panel ({\bf B}) shows the resultant -- optimal given the constraint -- values of OLE information $I_\OLE$, with different noise correlation strengths (blue circles). The strength of correlations is quantified by the Euclidean norm (Eq.~\eqref{E:rho_2norm}).  For comparison, the red crosses show information obtained for correlations that obey the sign rule (in particular, pointing along the gradient giving greatest information for weak correlations); this information is always less than or equal to the optimum, as it must be.  Note that correlations that follow the sign rule fail to exist for large correlation strengths, as the defining vector points outside of the allowed region (spectrahedron) beyond  a critical length (labeled (ii)).  For correlation strengths beyond this point, distinct optimized noise correlations continue to exist; the information values they obtain eventually saturate at noise-free levels (see text), which is $1$ for the example shown here. This occurs for a wide range of correlation strengths.  Panel ({\bf C}) shows how well these optimized noise correlations are predicted from the corresponding signal correlations (by the sign rule), as quantified by the $R^2$ statistic (between 0 and 1, see Fig.~\ref{F:scatter_C}). For small magnitudes of correlations, the $R^2$ values are high, but these decline when the noise correlations are larger. }
\label{F:diff_tuning}
\end{figure}

In the limit where the correlations are constrained to be small, the optimized correlations agree with the sign rule; an example of these ``local" optimized correlations is shown in Fig.~\ref{F:scatter_C}{\bf ADG}, corresponding to the point labeled  $(i)$ in Fig.~\ref{F:diff_tuning}{\bf BC}. This is predicted by Theorem~\ref{TH:p_or_n}.  In this ``local" region of near-zero noise correlations, we see a linear alignment of signal and noise correlations (Fig.~\ref{F:scatter_C}{\bf D}).  As larger correlation strengths are reached (points $(ii)$ and $(iii)$ in Fig.~\ref{F:diff_tuning}{\bf BC}), we observe a gradual violation of the sign rule for the optimized noise correlations.  This is shown by the gradual loss of the linear relationship between signal and noise correlations in Fig.~\ref{F:diff_tuning}{\bf D} vs.{\bf E} vs.{\bf F}, as quantified by the $R^2$ statistic. Interestingly, this can happen even when the correlation coefficients continue have reasonably small values, and are broadly consistent with the ranges of noise correlations seen in physiology experiments~\cite{Cohen:2011eh,hansen12,Ecker:2011bx}.

The two different regimes of optimized noise correlations arise because, at a certain correlation strength, the correlation strength can no longer be increased along the direction that defines the sign rule without leaving the region of positive semidefinite covariance matrices.  However, correlation matrices still exist that allow for more informative coding with larger correlation strengths. This reflects the geometrical shape of the spectrahedron, wherein the optima may lie in the ``corners", as shown in Fig.~\ref{F:3d_boundary}. For these larger-magnitude correlations, the sign rule no longer describes optimized correlations, as shown with an example of optimized correlations in Fig.~\ref{F:scatter_C}{\bf CF}.

Fig.~\ref{F:scatter_C} illustrates another interesting feature.  There is a diverse set of correlation matrices, with different Euclidean norms beyond the value of (roughly) 1.2, that all achieve the same globally optimal information level. As we see in the next section, this phenomenon is actually typical for large populations, and can be described precisely.

\begin{figure}[H]
\begin{center}
\includegraphics[width=3in]{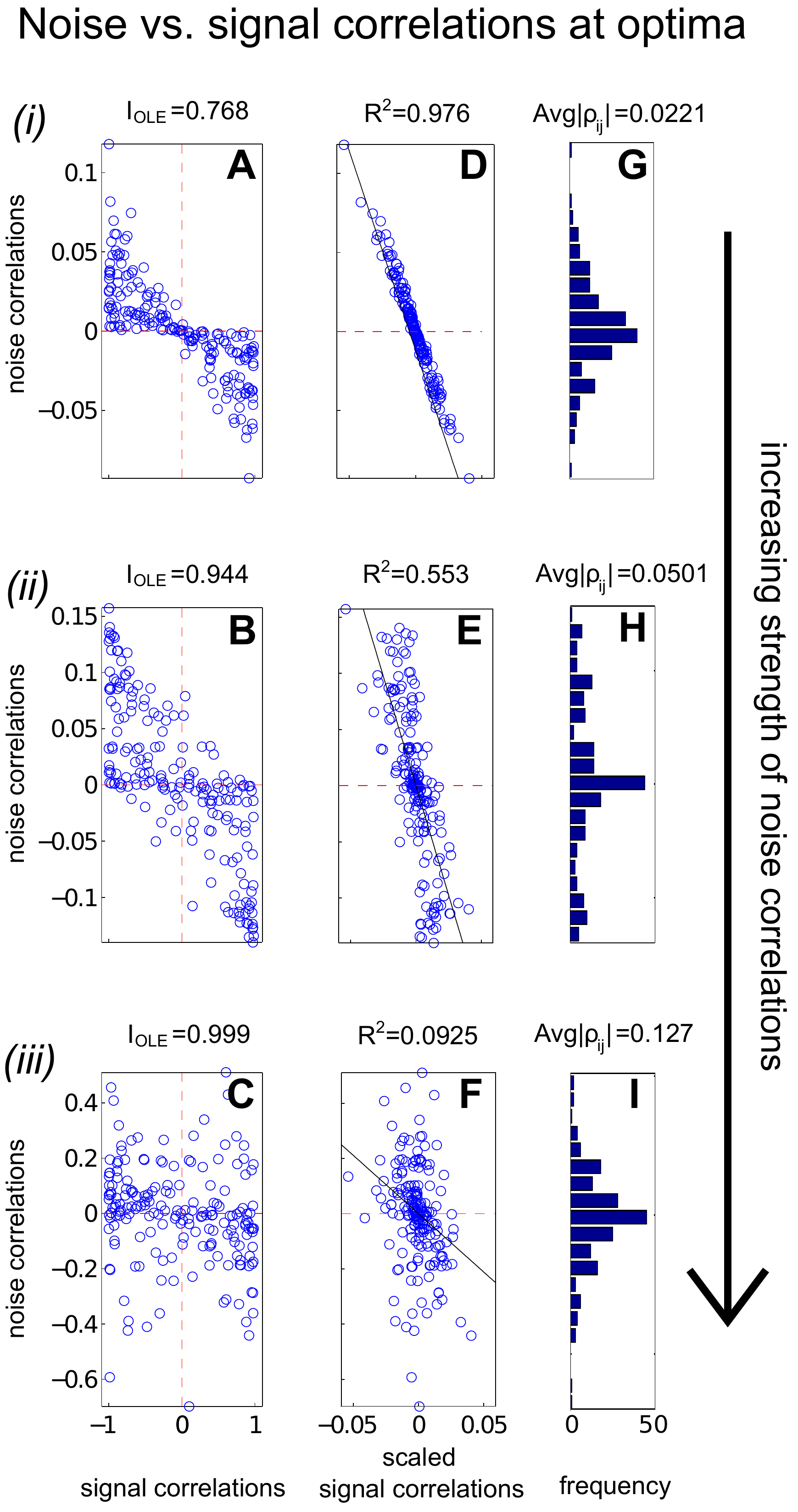}
\end{center}
\caption{\textbf{In our larger neural population, the sign rule governs optimal noise correlations only when these correlations are forced to be very small in magnitude; for stronger correlations, optimized noise correlations have a diverse structure.}  
Here we investigate the structure of the optimized noise correlations obtained in Fig.~\ref{F:diff_tuning}; we do this for three examples with increasing correlation strength, indicated by the labels $(i), \, (ii), \, (iii)$ in that figure.  First ({\bf ABC}) show scatter plots of the noise correlations of the neural pairs, as a function of their signal correlations (defined in Methods Section~\namerefquote{S:info_def}). For each example, we also show ({\bf DEF}) a version of the scatter plot where the signal correlations have been rescaled in a manner discussed in Section~\namerefquote{S:numerics}, that highlights the linear relationship (wherever it exists) between signal and noise correlations.  In both sets of panels, we see the same key effect:  the sign rule is violated as the (Euclidean) strength of noise correlations increases.  In ({\bf ABC}), this is seen by noting the quadrants where the dots are located: the sign rule predicts they should only be in the second and fourth quadrants. In ({\bf DEF}), we quantify agreement with the sign rule by the $R^2$ statistic. Finally, ({\bf GHI}) display histograms of the noise correlations; these are concentrated around 0, with low average values in each case.}
\label{F:scatter_C}
\end{figure}

\subsection*{Noise cancellation}
\label{S:noise_cancel}

For certain choices of tuning curves and noise variances, including the examples in Fig.~\ref{F:3d_boundary}{\bf A} and Section~\namerefquote{S:heter_eg}, we can tell precisely the value of the globally optimized information quantities --- that is, the information levels obtained with optimal noise correlations. For the OLE, this global optimum is the upper bound on $I_{\OLE}$. This is shown formally in Lemma \ref{TH:0noise}, but it simply translates to an intuitive lower bound of the OLE error, similar to the data processing inequality for mutual information. This bound states that the OLE error cannot be smaller than the OLE error when there is no noise in the responses, i.e. when the neurons produce a deterministic response conditioned on the stimulus. This upper bound may --- and often will (Theorem~\ref{TH:0noise_prob}) --- be achievable by populations of noisy neurons.  

Let us first consider an extremely simple example. Consider the case of two neurons with identical tuning curves, so that their responses are $x_i = \mu(s) + n_i$, where $n_i$ is the noise in the response of neuron $i \in \{1,2\}$, and $\mu(s)$ is the same mean response under stimulus $s$. In this case, the ``noise free" coding is when $n_1= n_2\equiv 0$ on all trials, and the inference accuracy is determined by the shape of the tuning curve $\mu(s)$ (whether or not it is invertible, for example). Now let us consider the case where the noise in the neurons' responses is non-zero but perfectly anti correlated, so that $n_1 = - n_2$ on all trials. We can then choose the read-out as $(x_1 + x_2)/2 = \mu(s)$ to cancel the noise and achieve the same coding accuracy as the ``noise free" case.

The preceding example shows that, at least in some cases, one can choose noise correlations in such a way that a linear decoder achieves ``noise-free" performance. 
One is naturally left to wonder whether this observation applies more generally.

First, we state the conditions on the {noise covariance matrices} under which the noise-free coding performance is obtained. We will then identify the conditions on {parameters of the problem}, i.e. the tuning curves (or receptive fields) and noise variances, under which this condition can be satisfied. Recall that the OLE is based on a fixed (across stimuli) linear readout coefficient vector $A$ defined in Eq.~\eqref{E:A_def}

\begin{restatable}{theorem}{ZeroNoiseCond}
\label{TH:0noise_cond}
A covariance matrix $C^n$ attains the noise-free bound for OLE information (and hence is optimal), if and only if $C^n A=C^n  (C^\mu)^{-1}L=0$. Here $L$ is the cross-covariance between the stimuli responses (Eq.~\eqref{E:def_L}), $C^\mu$ is the covariance of the mean response (Eq.~\eqref{E:def_Gamma_Cmu}), and $A$ is the linear readout vector for OLE, which is the same as in the noise-free case --- that is, $A=(C^n+C^\mu)^{-1}L=(C^\mu)^{-1}L$ --- when the condition is satisfied.
\end{restatable}

We note that when the condition is satisfied, the conditional variance of the OLE is $A^T C^n A=0$. This indicates that all the error comes from the bias, if we as usual write the mean square error (for scalar $s$) in two parts, $\EVb{(\hat{s}-s)^2}=\EVb{\var(\hat{s} | s)}+\var(\EVb{\hat{s}|s}-s)$. The condition obtained here can also be interpreted as ``signal/readout being orthogonal to the noise." While this perspective gives useful intuition about the result, we find that other ideas are more useful for constructing proofs of  this and other results. We discuss this issue more thoroughly in Section~\namerefquote{S:eigenvector_discuss}.

In general, this condition may not be satisfied by some choices of pairwise correlations. The above theorem implies that, given the tuning curves, the issue of whether or not such ``noise free" coding is achievable will be determined only by the relative magnitude, or heterogeneity,  of the noise variances for each neuron -- the diagonal entries of $C^n$. The following theorem outlines precisely the conditions under which such ``noise-free" coding performance is possible, a condition that can be easily checked for given parameters of a model system, or for experimental data.  

\begin{restatable}{theorem}{ZeroNoiseCondDiag}
\label{TH:0noise_cond_diag}
For scalar stimulus, let $q_i=\sqrt{A_i^2 C^n_{ii}}$, $i=1,\cdots, N$, where $A=(C^\mu)^{-1}L$ is the readout vector for OLE in the noise-free case. Noise correlations may be chosen so that coding performance matches that which could be achieved in the absence of noise if and only if
\beq
\label{E:max_cond}
\max \{q_i\} \leq \frac{1}{2} \sum_{i=1}^N q_i.
\eeq
When ``$<$" is satisfied, all optimal correlations attaining the maximum form a $\frac{N(N-3)}{2}$ dimensional convex set on the boundary of the spectrahedron. When ``$=$" is attained, the dimension of that set is $\frac{N_0(N_0+1)}{2}$, where $N_0$ is the number of zeros in $\{q_i\}$.  
\end{restatable}

We pause to make three observations about this Theorem.  First, the set of optimal correlations, when it occurs, is high-dimensional.  This bears out the notion that there are many different, highly diverse noise correlation structures that all give the same (optimal) level of the information metrics.  Second, and more technically, we note that the (convex) set of optimal correlations is flat (contained in a hyperplane of its dimension), as viewed in the higher dimensional space $\mathbb{R}^{\frac{N(N-1)}{2}}$. A third intriguing implication of the theorem is that when noise-cancellation is possible, all optimal correlations are connected, as the set is convex (any two points are connected by a linear segment that also in the set), and thus the case of disjoint optima as in Fig.~\ref{F:3d_boundary}{\bf B} will never happen when optimal coding achieves noise-free levels. Indeed, in Fig.~\ref{F:3d_boundary}{\bf B}, the noise-free bound is not attained. 

The high dimension of the convex set of noise-canceling correlations explains the diversity of optimal correlations seen in Fig.~\ref{F:diff_tuning}{\bf B} (i.e., with different Euclidean norms). Such a property is nontrivial from a geometric point of view. One may conclude prematurely that the dimension result is obvious if one considers algebraically the number of free variables and constraints in the condition of Theorem~\ref{TH:0noise_cond}. This argument would give the dimension of the resulting linear space. However, as shown in the proof, there is another nontrivial step to show that the linear space has some finite part that also satisfies the positive semidefinite constraint. Otherwise, many dimensions may shrink to zero in size, as happens at the corner of the spectrahedron, resulting in a small dimension.

The optimization problem can be thought of as finding the level set of information function associated with as large as possible value while still intersecting with the spectrahedron. The level sets are collections of all points where the information takes the same value.  These form high dimensional surfaces, and contain each other, much as layers of an onion. Here these surfaces are also guaranteed to be convex as the information function itself is. Next, note from Fig.~\ref{F:3d_boundary} that we have already seen that the spectrahedron has sharp corners. Combining this with our view of the level sets, one might guess that the set of optimal solutions --- i.e. the intersection --- should be very low dimensional. Such intuition is often used in mathematics and computer science, e.g. with regards to the sparsity promoting tendency of L1 optimization. The high dimensionality shown by our theorem therefore reflects a nontrivial relationship between the shape of the spectrahedron and the level sets of the information quantities.

Although our theorem only characterizes the abundance of the set of \emph{exact} optimal noise correlations, it is not hard to imagine the same, if not more, abundance should also hold for correlations that approximately achieve the maximal information level. This is indeed what we see in numerical examples.  For example, note the long, curved level-set curves in Fig.~\ref{F:2d_sign} near the boundaries of the allowed region.  Along these lines lie many different noise correlation matrices that all achieve the same nearly-optimal values of $I_{\OLE}$.  The same is true of the many dots in Fig.~\ref{F:3d_boundary}{\bf A} that all share a similar ``bright" color corresponding to large $I_{\OLE}$.

One may worry that the noise cancellation discussed above is rarely achievable, and thus somewhat spurious. The following theorem suggests that the opposite is true.  In particular, it gives one simple condition under which the noise cancellation phenomenon, and resultant high-dimensional sets of optimal noise correlation matrices, will almost surely be possible in large neural populations.  

\begin{restatable}{theorem}{ZeroNoiseProb}
\label{TH:0noise_prob}
If the $\{q_i\}$ defined in Theorem \ref{TH:0noise_cond_diag} are independent and identically distributed (i.i.d.) as a random variable $X$ on $[0, \infty)$ with $0<\EVb{X}<\infty$, then the probability
\beq
P(\text{the inequality of Eq.\eqref{E:max_cond} is satisfied}) \rightarrow 1 \text{, as } N \rightarrow \infty. 
\eeq
\end{restatable}

\noindent In actual populations, the $q_i$ might not be well described as i.i.d..  However, we believe that the conditions of the inequality of Eq.\eqref{E:max_cond} are still likely to be satisfied, as the contrary seems to require one neuron with highly outlying tuning and noise variance value (a few comparable outliers won't necessary violate the condition, as their magnitudes will enter on the right hand side of the condition, thus the condition only breaks with a single ``outlier of outliers").

\section*{Discussion}

\subsection*{Summary}

In this paper, we considered a general mathematical setup in which we investigated how coding performance changes as noise correlations are varied. Our setup made no assumptions about the shapes (or heterogeneity) of the neural tuning curves (or receptive fields), or the variances in the neural responses. Thus, our results -- which we summarize below -- provide general insights into the problem of population coding. These are as follows:

\bei
\item We proved that the {\it sign rule} --- if signal and noise correlations have opposite signs, then the presence of noise correlations will improve encoded information vs. the independent case --- holds for any neural population. In particular, we showed that this holds for three different metrics of encoded information, and for arbitrary tuning curves and levels of heterogeneity. Furthermore, we showed that, in the limit of weak correlations, the sign rule predicts the optimal structure of noise correlations for improving encoded information. 
\item  However, as also found in the literature (see below), the sign rule is not a necessary condition for good coding performance to be obtained. We observed that there will typically be a diverse family of correlation matrices that yield good coding performance, and these will often violate the sign rule.
\item There is significantly more structure to the relationship between noise correlations and encoded information than that given by the sign rule alone.  The information metrics we considered are all {\it convex} functions with respect to the entries in the noise correlation matrix.  Thus, we proved that the optimal correlation structures must lie on boundaries of any allowed set. These boundaries could come from mathematical constraints -- all covariance matrices must be positive semidefinite -- or mechanistic/biophysical ones.
\item Moreover, boundaries containing optimal noise correlations have several important properties.  First, they typically contain correlation matrices that lead to the same high coding fidelity that one could obtain in the absence of noise.  Second, when this occurs there is a high-dimensional set of different correlation matrices that all yield the same high coding fidelity -- and many of these matrices strongly violate the sign rule.  

\item Finally, for reasonably large neural populations, we showed that both the noise-free, and more general SR-violating optimal, correlation structures emerge while the average noise correlations remain quite low --- with values comparable to some reports in the experimental literature.
\eei

\subsection*{Convexity of information measures}
\label{S:convexity}

Convexity of information with respect to noise correlations arises conceptually throughout the paper, and specifically in Theorem~\ref{TH:boundary}. We have shown that such convexity holds for all three particular measures of information studied above ($I_{\Flin}$, $I_{\OLE}$, and $I_{\mutG}$).  Here, we show that these observations may reflect a property intrinsic to the concept of information, so that our results could apply more generally.

It is well known that mutual information is convex with respect to conditional distributions.  Specifically, consider two random variables (or vectors) $\vec{x},\vec{y}$, each with conditional distribution $\vec{x}|\vec{s} \sim  p_1(\vec{x} | \vec{s})$ and $\vec{y}|\vec{s} \sim  p_2(\vec{y} | \vec{s})$ (with respect the random ``stimulus'' variable(s) $\vec{s}$). Suppose another variable $\vec{z}$ has a conditional distribution given by a nonnegative linear combination of the two, $\vec{z}|\vec{s}\sim  p(\vec{z}|\vec{s})= \alpha p_1(\vec{z}|\vec{s}) + (1-\alpha)p_2(\vec{z}|\vec{s})$, $\alpha\in [0,1]$. The mutual information must satisfy $I(\vec{z},\vec{s})\leq  \alpha I(\vec{x},\vec{s})+ (1-\alpha)I(\vec{y},\vec{s})$. Notably, this fact can be proved using only the axiomatic properties of mutual information (the chain rule for conditional information and nonnegativity) \cite{Cover:2006ub}.

It is easy to see how this convexity in conditional distributions is related to the convexity in noise correlations we use. To do this, we further assume that the two conditional means are the same, $\EVb{\vec{x} |\vec{s}}=\EVb{\vec{y} | \vec{s}}$, and let $\vec{x}, \vec{y}$ be random vectors. Introduce an auxiliary Bernoulli random variable $T$ that is independent of $\vec{s}$, with probability $\alpha$ of being 1. The variable $\vec{z}$ can then be explicitly constructed using $T$: for any $\vec{s}$, draw $\vec{z}$ according to $p_1(\vec{z}|\vec{s})$ if $T=1$ and according to $p_2(\vec{z}|\vec{s})$ otherwise. Using the law of total covariance, the covariance (conditioned on $\vec{s}$) between the $i,j$ elements of $\vec{z}$ is
\beqrn
\cov_{\vec{s}}(z_i,z_j)&=&\EVb {\cov_{\vec{s}}(z_i,z_j | T)}+\cov_{\vec{s}}(\mathbf{E}_{\vec{s}}\{ z_i | T\}, \mathbf{E}_{\vec{s}}\{ z_j | T\})\\
&=& \alpha \cov_{\vec{s}}(z_i, z_j | T=1) +(1-\alpha) \cov_{\vec{s}}(z_i, z_j |T=0 )+ 0\\
&=&\alpha \cov_{\vec{s}}(x_i, x_j) +(1-\alpha) \cov_{\vec{s}}( y_i,y_j).
\eeqrn
This shows that the noise covariances are expressed accordingly as linear combinations. If the information depends only on covariances (besides the fixed means), as for the three measures we consider, the two notions of convexity become equivalent. A direct corollary of this argument is that the convexity result of Theorem~\ref{TH:boundary} also holds in the case of mutual information for conditionally Gaussian distributions (i.e., such that $\vec{x}$ given $\vec{s}$ is Gaussian distributed).

\subsection*{Sensitivity and robustness of the impact of correlations on encoded information}
\label{S:robust}
One obvious concern about our results, especially those related to the ``noise-free" coding performance, is that this performance may not be robust to small perturbations in the covariance matrix -- and thus, for example, real biological systems might be unable to exploit noise correlations in signal coding. This issue was recently highlighted, in particular, by~\cite{Beck:2013}.  

At first, concerns about robustness might appear to be alleviated by our observation that there is typically a large set of possible correlation structures that all yield similar (optimal) coding performance (Theorem~\ref{TH:0noise_cond_diag}). However, if the correlation matrix was perturbed along a direction orthogonal to the level set of the information quantity at hand, this could still lead to arbitrary changes in information. To address this matter directly, we explicitly calculated the following upper bound for the sensitivity of information, or \emph{condition number} $\kappa$ with respect (sufficiently small) perturbations. The condition number $\kappa$ is defined as the ratio of relative change in the function to that in its variables. For example, the condition number corresponding to perturbing $C^n$ is the smallest number $\kappa_{\Flin : C^n}$ that satisfies $\frac{ | \Delta I_{F,lin} |}{|I_{F,lin} |}\leq \kappa_{\Flin:  C^n} \frac{ \norm{ \Delta C^n}}{ \norm{ C^n}}$. Similarly one can define condition number $\kappa_{\Flin: \nabla \mu}$ for perturbing the tuning of neurons $\nabla \mu$.

\begin{restatable}{proposition}{CondCn}
\label{TH:cond_Cn}
The local condition number of $I_{\Flin}$ under perturbations of $C^n$ (where magnitude is quantified by 2-norm) is bounded by 
\beq
\label{E:cond_Cn}
\kappa_{\Flin : C^n}\leq
2 \kappa_2(C^n):= 2\tnorm{ (C^n)^{-1}}\cdot  \tnorm{ C^n}=\frac{2\lambda_{\max}}{\lambda_{\min} },
\eeq
where $\lambda_{\max}$ and $\lambda_{\min}$ are the largest and smallest eigenvalue of $C^n$ respectively. Here $\kappa_2$ is the condition number with respect to the 2-norm, as defined in the above equation.

Similarly, the condition number for perturbing of $\nabla \mu$ is bounded by
 \beq
 \label{E:cond_Cn_L}
 \kappa_{\Flin : \nabla \mu}\leq
 3 \sqrt{\kappa_2(C^n) K}\frac{ \max_i\tnorm{ (\nabla \mu)_{\cdot,i} }}{\min_i \tnorm{ (\nabla \mu)_{\cdot,i} }},
 \eeq
 where $(\nabla \mu)_{\cdot,i}$ is the $i$-th column of $\nabla \mu$ and assume $(\nabla \mu)_{\cdot,i}\neq 0$ for all $i$. Here $K$ is the dimension of the stimulus $s$.
\end{restatable} 
Though stated for $I_{\Flin}$, same results also hold for $I_\OLE$ when replacing $C^n$ by $C^\mu+C^n$ in Eq.~\eqref{E:cond_Cn} and \eqref{E:cond_Cn_L}. We believe that a similar property is possible to derive for mutual information $I_\mutG$, but that the expression could be quite cumbersome; we do not pursue this further here.

To interpret this Proposition, we make the following observations which explain when the sensitivity or condition numbers will (or will not) be themselves reasonable in size, for given noise correlations $C^n$. 
In our setup, the diagonal of $C^n$ (or $C^\mu+C^n$ for OLE) is fixed, and therefore $\lambda_{\max}$ is bounded (Gershgorin circle theorem). As long as $C^n$ (or $C^\mu+C^n$) is not close to singular, the information should therefore be robust, i.e. with a reasonably bounded condition number. For OLE, as $C^\mu+C^n \succcurlyeq C^\mu$, we always have a universal bound of $\kappa$ determined only by $C^\mu$. For the linear Fisher information, however, nearly singular $C^n$ can more typically occur near optimal solutions; in these cases, the condition numbers will be very large.

\subsection*{Relationship to previous work}
Much prior work has investigated the relationship between noise correlations and the fidelity of signal coding~\cite{Zohary:1994ei,Averbeck:2006ew,Abbott:1999ul,Ecker:2011bx, shamir06, Sompolinsky:2001hh, Cohen:2011eh,wilke02,Josic:2009du,Averbeck:2006vj,romo03,daSilveira:2013vf}. Two aspects of our current work complement and generalize those studies.

The first are our results on the sign rule (Section~\namerefquote{S:sign_rule}).  Here, we find that, if each cell pair has noise correlations that have the opposite sign vs. their signal correlations, the encoded information is always improved, and that, at least in the case of weak noise correlations, noise correlations that have the same sign as the signal correlations will diminish encoded information. This effect was observed by~\cite{Zohary:1994ei} for neural populations with identically tuned cells.  Since the tuning was identical in their work, all signal correlations were positive. Thus, their observation that positive noise correlations diminish encoded information is consistent with the SR results described above. 

Relaxing the assumption of identical tuning, several studies followed~\cite{Zohary:1994ei} that used cell populations with tuning that differed from cell to cell, but maintained some homogeneous structure -- i.e., identically shaped, and evenly spaced (along the stimulus axis) tuning curves, e.g., \cite{Averbeck:2006ew,Abbott:1999ul}.  The models that were investigated then assumed that the noise correlation between each cell pair was a decaying function of the displacement between the cells' tuning curve peaks. The amplitude of the correlation function -- which determines the maximal correlation over all cell pairs, attained for ``nearby" cells -- was the independent variable in the numerical experiments. Recall that these nearby (in tuning-curve space) cells, with overlapping tuning curves, will have positive signal correlations.  These authors found that positive signs of noise correlations diminished encoded information, while negative noise correlations enhanced it.  This is once again broadly consistent with the sign rule, at least for nearby cells which have the strongest correlation.  Finally, we note that \cite{Sompolinsky:2001hh, Averbeck:2006ew,Roudi_Latham} give a crisp geometrical interpretation of the sign rule in the case of $N=2$ cells. 

At the same time, experiments typically show noise correlations that are stronger for cell pairs with higher signal correlations~\cite{Zohary:1994ei, Cohen:2011eh, kohn05}, which is certainly not in keeping with the sign rule. This underscores the need for new theoretical insights. To this effect, we demonstrated that, while noise correlations that obey the sign rule are guaranteed to improve encoded information relative to the independent case, this improvement can also occur for a diverse range of correlation structures that violate it.  (Recall the asymmetry of our findings for the sign rule:  noise correlations that violate the sign rule are only guaranteed to diminish encoded information if those noise correlations are very weak).  

This finding is anticipated by the work of \cite{Ecker:2011bx, shamir06,daSilveira:2013vf}, who used elegant analytical and numerical studies to reveal improvements in coding performance in cases where the sign rule was violated.  They studied heterogeneous neural populations, with, for example, different maximal firing rates for different neurons.  In particular, these authors show how heterogeneity can simultaneously improve the accuracy and capacity of stimulus encoding ~\cite{daSilveira:2013vf}, or can create coding subspaces that are nearly orthogonal to directions of noise covariance~\cite{Ecker:2011bx, shamir06}.   Taken together, these studies show that the same noise correlation structure discussed above -- with nearby cells correlated -- could lead to improved population coding, so long as the noise correlations are sufficiently strong.  \cite{Ecker:2011bx} also demonstrated that the magnitude of correlations needed to satisfy the ``sufficiently strong" condition decreases as the population size increases, and that in the large $N$ limit, certain coding properties become invariant to the structure of noise correlations. Overall, these findings agree with our observations about a large diversity of SR-violating noise correlation structures that improve encoded information. 

One final study requires its own discussion. Whereas the current study (and those discussed above) investigated how coding relates to noise correlations with no concerns for the biophysical origin of those correlations, \cite{tkacik2010} studied a semi-mechanistic model in which noise correlations were generated by inter-neuronal coupling. They observed that coupling that generates anti-SR correlations is beneficial for population coding when the noise level is very high, but that at low noise levels, the optimal population would follow the SR. Understanding why different mechanistic models can display different trends in their noise correlations is important, and we are currently investigating that issue.

\subsection*{The geometry of the covariance matrix}
\label{S:eigenvector_discuss}

One geometrical, and intuitively helpful, way to think about problems involving noise correlations is to ask when the noise is ``orthogonal to the signal": in these cases, the noise can be separated from or be orthogonal to the signal, and high coding performance is obtained.  This geometrical view is equally valid for the cases we study (e.g., the conditions we derive in Theorem~\ref{TH:0noise_cond}), and is implicit in the diagrams in Figure~\ref{F:tuning}.  To make the approach explicit, one could perform an eigenvector analysis on the covariance matrices at hand, where quantities like linear Fisher information are rewritten as a sum of projections of the tuning vector to the eigen-basis of the covariance matrix, weighted by the appropriate eigenvalues.  

This invites the question of whether a simpler way to obtain the results in our paper wouldn't be to consider how covariance eigenvectors and eigenvalues could be manipulated more directly. For example, if one could simply ``rotate" the eigenvectors of the covariance matrix out of the signal direction -- or shrink the eigenvalues in that direction -- one would necessarily improve coding performance. So why don't we simply do this when exploring spaces of covariance matrices? The reason is that these eigenvalue and eigenvector manipulations are not as easy to enact as they might at first sound (to us, and possibly to the reader). Recall that we asked how noise correlations affect coding subject to the specific constraint that the noise variance of each neuron is fixed, which translates in general to rather complex constraints on the eigenvalues and eigenvectors. For example, the eigenvalues of a fixed-diagonal covariance matrix  cannot be \emph{equivalently} described by simply having a fixed sum (which is a necessary condition for the diagonals to be constant, but is not a sufficient one). These facts limit the insights that a direct approach to adjusting eigenvalues and eigenvectors can have for our problem, and emphasize the non-trivial nature of our results.

An exception comes, for example, in special cases when the covariance matrix has a circulant structure, and consequently always has the Fourier basis for eigenvectors. These cases include many situations considered in the literature~\cite{Sompolinsky:2001hh,Ecker:2011bx}.  For contrast, the covariance matrices we studied were allowed to change freely, as long as the diagonals remained fixed.

\subsection*{Limitations and extensions}
We have developed a rich picture of how correlated noise impacts population coding.  For our results on noise cancellation in particular, this was done by allowing noise correlations to be chosen from the largest mathematically possible space (i.e., the entire spectrahedron). This describes the fundamental structure of the problem at hand, but are conclusions derived in this way important for biology? It is not hard to imagine many biological constraints that may further limit the range of possible noise correlations (e.g., limits on the strength of recurrent connections or shared inputs). On the one hand, the likelihood that the underlying phenomena could be found in biological systems seems increased by the fact that many different correlation matrices will suffice for noise free coding and that, as we discuss in Proposition~\ref{TH:cond_Cn}, information levels appear to have some robustness under perturbations of the underlying correlation matrices. 

However, care must still be taken in interpreting what we mean by ``noise free."  As emphasized by, e.g., ~\cite{Beck:2012vm,Ecker:2011bx}, noise upstream from the neural population in question can never be removed in subsequent processing. Therefore, the ``noise free" bound we discuss in Lemma~\ref{TH:0noise} should not allow for a higher information level than that determined by this upstream noise.  In some cases, this fact could lead to a consistency requirement on either the set of signal correlations $C^\mu$, the set of allowed noise correlations $C^n$, or both. To specify these constraints and avoid possible over-interpretations of the abstract coding model as we study, one could combine a explicit mechanistic model with the present approach.  

On another note, we have asked what noise correlations allow for linear decoders to best recover the stimulus from the set of neural population responses. At the same time, there is reason to be wary of linear decoders~\cite{Shamir:2004bw} (see also~\cite{Josic:2009du}), as they might miss significant information that is only accessible via a non-linear read-out. Furthermore, given the non-linearity inherent in dendritic processing and spike generation~\cite{KochBP}, there is added motivation to consider information without assuming linearity.

Furthermore, we have herein restricted ourselves to asking about pairwise noise correlations, while there are many studies that identify higher-order correlations (HOC) in neural data~\cite{Ganmor:2011ct,ohiorhenuan10}, and some numerical results~\cite{Zylberberg:2012ty} that hint at when those HOC are beneficial for coding. In light of this study, it is interesting to ask whether we can derive a similarly general theory for HOC, and to investigate how the optimal pairwise and higher-order correlations interrelate. Note this issue is closely related to the type of decoder that is assumed: the performance of linear decoder (as measured by mean squared error) depends on the pairwise correlations, but not HOC. Therefore the effect of HOC must be studied in the context of nonlinear coding.

Finally, we note that here we used an abstract coding model that evaluates information based on the statistics $C^n, C^\mu, L$ and so on. For generality, we made no assumptions on the structure of these statistics, and any links among them. This suggests two questions for future work: whether an arbitrary set of such statistics is realizable in a constructive model of random variables, and whether there are any typical relationships between these statistics when they arise from tuned neural populations.  As a preliminary investigation, we partially confirmed the answer to the first question, except for a ``zero measure'' set of statistics, under generic assumptions (data not shown).

\subsection*{Experimental implications}

Recall that we observed that, in general, for a given set of tuning curves and noise variances, there will be a diverse family of noise correlation matrices that will yield good (optimal, or near-optimal) performance. This effect can be observed in Figs.~\ref{F:2d_sign}, \ref{F:3d_boundary}, and \ref{F:scatter_C}, as well as our result about the dimension of the set of correlation matrices that yield (when it is possible) noise-free coding performance (Theorem~\ref{TH:0noise_cond_diag}).

At least compared with the alternative of a unique optimal noise correlation structure, our findings imply that it could be relatively ``easy" for the biological system to find good correlation matrices. At the same time, since the set of good solutions is so large, we should not be surprised to see heterogeneity in the correlation structures exhibited by biological systems.  Similar observations have previously been made in the context of neural oscillators: Prinz and colleagues~\cite{prinz04} observed that neuronal circuits with a variety of different parameter values could produce the types of rhythmic activity patterns displayed by the crab stomatogastric ganglion. Consequently, there is much animal-to-animal variability in this circuit~\cite{marder11}, even though the system's performance is strongly conserved.

At the same time, the potential diversity of solutions could present a serious challenge for analyzing data (cf.~\cite{Beck:2013}). Notice, that, at least for the $N=3$ cases of Figs.~\ref{F:2d_sign} and~\ref{F:3d_boundary} for example, how much the performance can vary as one of the correlation coefficients is changed, while keeping the other ones fixed. If this phenomenon is general, it means that, in an experiment where we observe a (possibly small) subset of the correlation coefficients, it may be very hard to know how those correlations actually affect coding: the answer to that question depends strongly on all of the other (unobserved) correlation coefficients. As our recording technologies improve~\cite{stevenson11}, and we make more use of optical methods, these ``gaps" in our datasets will get smaller, and this issue may be resolved; further theoretical work to gauge the seriousness of the underlying issue is also needed. In the meanwhile, caution seems wise when analyzing noise correlations in sparsely sampled data.

Finally, recall that the optimal noise correlations will always lie on the boundary of the allowed region of such correlations.  Importantly, what we mean by that boundary is flexible.  It may be the mathematical requirement of positive semidefinite covariance matrices -- the loosest possible requirement -- or there may be tighter constraints that restrict the set of correlation coefficients. Since biophysical mechanisms determine noise correlations, we expect that there will be identifiable regions of possible correlation coefficients that are possible in a given circuit/system. Understanding those ``allowed" regions will, we anticipate, be important for attempts to relate noise correlations to coding performance, and ultimately to help untangle the relationship between structure and function in sensory systems.

\section*{Methods}

In the Methods below, we will first revisit the problem set-up, and define our metrics of coding quality. We will then prove the theorems from the main text. Finally, we will provide the details of our numerical examples. A summary of our most frequently used notation is listed in Table~\ref{T:notation}.

\begin{table}[H]
\caption{
\bf{Notations}}
\begin{tabular*}{\textwidth}{ l @{\extracolsep{\fill}}  r}
\hline
$\vec{s}$ &  stimulus   \\
$x_i$ & response of neuron $i$ \\
$\mu_i$ & mean response of neuron $i$ \\
$\nabla \mu_i$ & derivative against $\vec{s}$, Eq.~\eqref{E:nabla_mu}\\
$L$ & covariance between $x$ and $\vec{s}$, Eq.~\eqref{E:def_L}\\
$C^n$ & noise covariance matrix (averaged or conditional, Section~\namerefquote{S:set-up})\\
$C^\mu$ & covariance of the mean response, Eq.~\eqref{E:def_Gamma_Cmu}\\
$\succ 0$, $\succcurlyeq 0$ & (a matrix is) positive definite and positive semidefinite\\
$\Gamma=C^\mu+C^n$ & total covariance,Eq.~\eqref{E:def_Gamma_Cmu}\\
$A=\Gamma^{-1} L$ & optimal readout vector of OLE, Eq.~\eqref{E:A_def}\\
$\rho_{ij}$ & noise correlations, Eq.~\eqref{E:def_rho}\\
$\rhosig_{ij}$ & signal correlations, Eq.~\eqref{E:rhosig_def}\\
$I_\Flin$ & linear Fisher information, Eq.~\eqref{E:Flin_def}\\
$I_\OLE$ & OLE information (accuracy of OLE), Eq.~\eqref{E:OLE}\\
$I_\mutG$ & mutual information for Gaussian distributions, Eq.~\eqref{E:mutG_def}\\
\hline
\end{tabular*}
\label{T:notation}
\end{table}

\subsection*{Summary of the problem set-up}
\label{S:set-up}

We consider populations of neurons that encode a stimulus $\vec{s}$ by their noisy responses $x_i$. For simplicity, we will suppress the vector notation in the Methods Unless otherwise stated, most of our results apply equally well to either scalar, or multi-dimensional, stimuli.

The mean activity or ``tuning" of the neurons are described by $\mu_i( s)=\EVb{x_i | s}$. In the case of scalar stimuli, this corresponds to the notion of a tuning curve. For more complex stimuli, this is more aligned with the idea of a receptive field.

The trial-to-trial noise part in $x_i$, given a fixed stimulus, can be described by the conditional covariance $C^n_{ij} = \cov \left(x_i, x_j  | s \right)$ (superscript $n$ denotes ``noise"). In particular $C^n_{ii}=\var\left(x_i | s \right)$ are noise variances of each neuron.

We ask questions of the following type: given fixed tuning curves $\mu$ and noise variances $C^n_{ii}$, how does the choice of noise covariance structure $C^n_{ij}$, $i\neq j$ affect linear Fisher information $I_{\Flin}$ (see Section~\namerefquote{S:info_def})?

Besides the local information measure $I_{\Flin}$ that quantifies coding near a specific stimulus, we also considered global measures that describe overall coding of the entire ensemble of stimuli. These are $I_\OLE$ and $I_\mutG$, described in Section~\namerefquote{S:info_def}. For these quantities, the relevant noise covariance is $\cov(x_i,x_j)=\EVb{\cov \left(x_i, x_j  | s \right)}$. We overload the notation with $C^n=\cov(x_i,x_j)$ in these global coding contexts. The optimization problem can then be identically stated for $I_\OLE$ and $I_\mutG$.

\subsection*{Defining the information quantities, signal and noise correlations}
\label{S:info_def}

\subsubsection*{Linear Fisher information}
Linear Fisher information quantifies how accurately the stimulus near a value $s$ can be decoded by a local linear unbiased estimator, and is given by

\beq
\label{E:Flin_def}
I_{\Flin} =\nabla \mu^T (C^n)^{-1} \nabla \mu.
\eeq
In the case of a $K$ dimensional stimulus the same definition holds, with

\beq
\label{E:nabla_mu}
\nabla \mu=
\left(
\barr{ccc}
\frac{\partial \mu_1}{\partial s_1} & \cdots & \frac{\partial \mu_1}{\partial s_K}\\
\vdots & \vdots & \vdots \\
\frac{\partial \mu_N}{\partial s_1} & \cdots & \frac{\partial \mu_N}{\partial s_K}
\earr
\right).
\eeq
In order for $I_\Flin$ to be defined by Eq.~\eqref{E:Flin_def}, we assume $C^n$ is invertible and hence positive definite: $C^n\succ 0$.
It can be shown that $I_{\Flin}^{-1}$ is the (attainable) lower bound of the covariance matrix of the error of any local linear unbiased estimator. Here the term lower bound is used in the sense of positive semidefiniteness, that is the ordering $A \succcurlyeq B$ if and only if $A-B  \succcurlyeq 0$. To obtain a scalar information quantity, we consider $\tr(I_{\Flin})$ and also denote this by $I_{\Flin}$ if not stated otherwise.

\subsubsection*{Optimal linear estimator}
To quantify the \emph{global} ability of the population to encode the stimulus (instead of \emph{locally}, as for discrimination tasks involving small deviations from a particular stimulus value), we follow \cite{Salinas:1994wr} and consider a linear estimator of the stimulus, given responses $x$:
\beq
\hat{s}=\sum_{i} A_i x_i+s_0=A^T x+s_0,
\eeq
with fixed parameters $A_i$ and $s_0$ unchanged with $s$. The set of readout coefficients $A$ that minimize the mean square error for a scalar random stimulus $s$, i.e.
\beq
\EVb{(\hat{s}-s)^2}  \; ,
\eeq
can be solved analytically as in \cite{Salinas:1994wr}, yielding:
\beq
\label{E:A_def}
A=\Gamma^{-1}L,\qquad \min (\EVb{(\hat{s}-s)^2})=\var(s)-L^T\Gamma^{-1} L,
\eeq
where 
\beq
\label{E:def_Gamma_Cmu}
(\Gamma)_{ij}=\cov(x_i,x_j)=\cov(\mu_i,\mu_j)+\EVb{\cov(x_i,x_j | s)}:=C^\mu+C^n,
\eeq
and $L$ is a column vector with entries $L_i=\cov(x_i,s)$. Here the expectation $\EVb{\cdot}$ generally means averaging over both noise and stimulus (except in $\EVb{\cov(x_i,x_j | s)}$, where averaging is only over the stimulus).

For multidimensional stimuli $s$, similar to the case for linear Fisher information, the lower bound (in sense of positive semidefiniteness) of the error covariance $\EVb{ (\hat{s} -s) (\hat{s} -s)^T}$ is given by $\cov(s,s)-L^T\Gamma^{-1} L$. Here $L$ is extended to form a matrix
\beq
\label{E:def_L}
 L=
 \left(
\barr{ccc}
\cov (x_1, s_1) & \cdots & \cov (x_1, s_K)\\
\vdots & \vdots & \vdots \\
\cov (x_N, s_1) & \cdots &\cov (x_N, s_K)
\earr
\right).
\eeq
Furthermore, a corresponding lower bound for the sum of squared errors $\EVb{ \tnorm{ \hat{s} -s }^2}$ is the scalar version $\tr(\cov(s,s))-\tr(L^T\Gamma^{-1} L)$.

When minimizing the OLE error with respect to noise correlations, $\mu_i$, $\cov(s,s)$ and $L$ are constants with respect to the optimization. Minimizing OLE error is therefore equivalent to maximizing the second term above, given by $L^T\Gamma^{-1} L$. This motivates us to define what we call ``the information for OLE'', which is simply the second term (above) --- i.e., the term that is subtracted from the signal variance to yield the OLE error.
Specifically, \beq
\label{E:OLE}
I_{\OLE}=L^T (C^\mu+C^n)^{-1}L \; \; \; \text{, or the scalar version  } \; \; \; \tr(L^T (C^\mu+C^n)^{-1}L).
\eeq
Thus, when $I_{\OLE}$ is large, the decoding error is small, and vice versa.
Comparing with the expression for $I_\Flin$, we see a similar mathematical structure, which will enable almost identical proofs of our theorems for both of these measures of coding performance.

Similar to $I_\Flin$, we need $C^\mu+C^n$ to be invertible in order to calculate $I_\OLE$. Since the signal covariance matrix $C^\mu$ does not change as we vary  $C^n$, this requirement is easy to satisfy. In particular, we assume $C^\mu$ is invertible ($C^\mu \succ 0$), and thus for all consistent  -- i.e. positive semidefinite -- $C^n$, $C^\mu+C^n\succcurlyeq C^\mu \succ 0$, so that  $C^\mu+C^n$ is invertible.

\subsubsection*{Mutual information for Gaussian distributions}

While the OLE and the linear Fisher information assume that a linear read-out of the population responses is used to estimate the stimulus, one may also be interested in how well the stimulus could be recovered by more sophisticated, nonlinear estimators. Mutual information, based on Shannon entropy is a useful quantity of this sort. It has many desirable properties consistent with the intuitive notion of ``information", and it we will use it to quantify how well a non-linear estimator could recover the stimulus.

Assuming that the joint distribution of $(x,s)$ is Gaussian ($s$ can be multidimensional), the mutual information has a simple expression

\begin{eqnarray}
\nonumber
I_{\mutG}&=&\frac{1}{2}\log\det(\cov(s,s))-\frac{1}{2}\log \det( \cov(s,s)- L^T \Gamma^{-1} L)\\
\label{E:mutG_def}
&=& \frac{1}{2}\log\det(\cov(s,s))-\frac{1}{2}\log \det( \cov(s,s)- L^T (C^\mu+C^n) ^{-1} L).
\end{eqnarray}
The quantities above are the same as in the definitions of $I_{\OLE}$.  Moreover, $\log$ is taken to base $e$, and hence the information is in units of nats. To convert to bits, one must simply divide our $I_{\mutG}$ values by $\log(2)$.

There is a consistency constraint that must be satisfied by any joint distribution of $(x,s)$, namely that
\beq
\label{E:mutG_cond}
\cov(s,s) -L^T (C^\mu+C^n)^{-1} L \equiv \cov(s,s | x) \succcurlyeq 0.
\eeq
This guarantees that $I_{\mutG}$ is always defined and real (but could be $+\infty$). To keep $I_{\mutG}$ finite, one needs to further assume $\cov(s,s) -L^T (C^\mu+C^n)^{-1} L \succ 0$, which is equivalent to $C^n \succ 0$. This can be seen by rewriting mutual information while exchanging the position of the two variables (since mutual information is symmetric),
\beqrn
I_\mutG&=&\frac{1}{2}\log\det(\cov(x,x))-\frac{1}{2}\log \det( \cov(x,x)- L \cov(s,s)^{-1} L^T)\\
&=&\frac{1}{2}\log\det(\cov(x,x))-\frac{1}{2}\log \det( \cov(x,x|s))=\frac{1}{2}\log\frac{\det(\cov(x,x))}{\det(C^n)}.
\eeqrn

It is easy to see that the formula contains terms similar to those in $I_\OLE$ and $I_\Flin$. In the scalar stimulus case, since $\log(\cdot)$ is an increasing function, maximizing $I_\mutG$ is equivalent to maximizing $I_\OLE$. In fact, the leading term in the Taylor expansion of $I_{\mutG}$ with respect to $L^T(C^\mu+C^n)^{-1}L$ is $\frac{L^T(C^\mu+C^n)^{-1}L}{2\var(s)}$, which is proportional to $I_{\OLE}$.  In the case of multivariate stimuli $s$, we note that the operation $\log \det(\cdot )$ preserves ordering defined in the positive semidefinite sense, i.e. $F \succcurlyeq G \Rightarrow \log \det(F ) \ge  \log \det(G ) $.  This close relationship suggests a way of transforming $I_{\OLE}$ to a comparable scale of information in nats (or bits) as $\frac{1}{2}\log\det(\cov(s,s))-\frac{1}{2}\log \det(\cov(s,s) - I_{\OLE})$.

\subsubsection*{Signal and noise correlations}
Given the noise covariance matrix $C^n$ one can normalize it as usual by its diagonal elements (variances) to obtain correlation coefficients
\beq
\label{E:def_rho}
\rho_{ij}=\frac{C^n_{ij}}{\sqrt{C^n_{ii} C^n_{jj}}}.
\eeq

We next discuss signal correlations, which describe how similar the tuning of a pair of neurons is. For linear Fisher information, we define signal correlations as
\beq
\label{E:rhosig_def}
\rhosig_{ij}=
\frac{\nabla \mu_i \cdot \nabla \mu_j}{\tnorm{ \nabla \mu_i  }  \tnorm{ \nabla \mu_j }}.
\eeq
Here $\nabla \mu_i=(\frac{\partial \mu_i}{\partial s_1},\cdots,\frac{\partial \mu_i}{\partial s_K})$ is the sensitivity vector describing how the mean response of neuron $i$ changes with $s$. With the above normalization, $\rhosig_{ij}$ takes value between $-1$ and $1$. 

For the other two information measures we use, $I_\OLE$ and $I_\mutG$, a similar signal correlation can be defined.  Here, we first define analogous tuning sensitivity vectors $A^0_{i}$ for each neuron, which will replace $\nabla \mu_i$ in Eq.~\eqref{E:rhosig_def}.  These vectors are 
\beq
\label{E:rhosig_def_OLE}
A^0=(C^\mu+D^n)^{-1}L \; \; \; \text{    and    }  \; \; \; A^0=(C^\mu+D^n)^{-1}LM^{-\frac{1}{2}}
\eeq
for $I_\OLE$ and $I_\mutG$ respectively. Here $D^n$ is the diagonal matrix of noise variances, and $M=\cov(s,s)- L^T(C^\mu+C^n)^{-1}L$.

The definitions of signal correlations above are chosen so that they are tied directly to the concept of the sign rule, as demonstrated in the proof of Theorem~\ref{TH:p_or_n}. As a consequence, for the case of $I_\OLE$ and $I_\mutG$, signal correlations are defined through the population readout vector.  This has an important implication that we note here.  Consider a case where only a subset of the total population is ``read out" to decode a stimulus.  Then, the population readout vector --- and hence the signal correlations defined above --- could vary in magnitude and even possibly change signs depending on which neurons are included in the subset.

A different definition of signal correlations for OLE is sometimes used in literature, which we denote by $\tilde{\rho}^{\text{sig}}_{ij}={C^\mu_{ij}}/{\sqrt{C^\mu_{ii}C^\mu_{jj}}}$. Naturally, one should not expect our sign rule results to apply exactly under this definition. However, when we redid our plots of signal vs. noise correlations using $\tilde{\rho}^{\text{sig}}_{ij}$ for our major numerical example (Fig.~\ref{F:scatter_C}{\bf ABC}), we observed the same qualitative trend (data not shown).  This reflects the fact that, at least in this specific example, the signal correlations defined in the two ways are positively correlated.  Understanding how general this phenomenon is would require further studies taking into account how the relevant statistics ($C^\mu$, L, etc.) are generated from tuning curves or neuron models.

We next define the notion of the magnitude or strength of correlations, as came up throughout the paper. In particular, in Section~\namerefquote{S:heter_eg}, we considered restrictions on the magnitudes of noise correlations when finding their optimal values.  We proceed as follows.  Since $\rho_{ij}=\rho_{ji}$, the list of all pairwise correlations of the population can be regarded as a single point in $\mathbb{R}^{\frac{N(N-1)}{2}}$.  If not stated otherwise, the vector 2-norm in that space (Euclidean norm) is what we call the ``strength of correlations:"
\beq
\label{E:rho_2norm}
\sqrt{\sum_{i<j} \rho_{ij}^2 }.
\eeq

\subsection*{Proof of Theorem \ref{TH:p_or_n}: the generality of the sign rule}
\label{S:SR_proof}

We will now restate and then prove Theorem \ref{TH:p_or_n}, first for $I_\Flin$ and then for $I_{\OLE}$ and $I_\mutG$.

\PoN*

The proof proceeds by showing that information increases along the direction indicated by the sign rule, and that the information quantities are convex, so that information is guaranteed to increase monotonically along that direction.

\begin{proof}
Consider linear Fisher information
\beq
I_{\Flin}=\tr(\nabla \mu^T (C^n)^{-1} \nabla \mu).
\eeq
Let $D^n$ be the diagonal part of $C^n$, corresponding to (noise) variance for each neuron. We change the off-diagonal entries of $C^n$ along a certain direction $(C^n)^\prime$ in $\mathbb{R}^{\frac{N(N-1)}{2}}$ and consider a parameterization of the resultant covariance matrix, with parameter $t$: $C^n(t)=D^n+(C^n)^\prime t$. We evaluate the directional derivative ($\frac{d}{dt}$) of $I_{\Flin}$ at $C^n=D^n$, 
\begin{eqnarray}
I_{\Flin}^\prime
\nonumber
&=&-\tr(\nabla \mu^T (D^n)^{-1} (C^n)^\prime(D^n)^{-1} \nabla \mu)\\
\nonumber
&=&-\tr((D^n)^{-1} (C^n)^\prime(D^n)^{-1} ( \nabla \mu \nabla \mu^T ))\\
\label{E:dir_prime}
&=& -2\sum_{i<j} \frac{(C^n)_{ij}^\prime \nabla \mu_i  \cdot \nabla \mu_j} { (D^n)_{ii} (D^n)_{jj}}.
\end{eqnarray}
Here $\nabla \mu_i=(\frac{\partial \mu_i}{\partial s_1},\cdots,\frac{\partial \mu_i}{\partial s_K})$, and we have used the identity $\tr(A B^T)=\sum_{ij}A_{ij} B_{ij}$ and the fact $\frac{dX^{-1}(t)}{dt}=-X^{-1}X^\prime X^{-1}$.  Recalling the definition of signal correlations in Eq.~\eqref{E:rhosig_def}, if the sign of $(C_{i,j}^n)^\prime$ is chosen to be opposite to the sign of $\rhosig_{ij}$  for all $i\neq j$, then Eq.~\eqref{E:dir_prime} ensures that the directional derivative $I_{\Flin}^\prime>0$ at $C^n=D^n$. 

We now derive a global consequence of this local derivative calculation.  $I_{\Flin}$ as a function of $t$ has $\frac{d I_\Flin}{dt}>0 |_{t=0} $. Since $I_\Flin$ is smooth, there exists $\delta>0$, such that for $t\in [0,\delta]$, $\frac{d I_\Flin (t)}{dt}>0$. For corresponding $C^n(t)$, applying the mean value theorem, we have $I_\Flin(C^n(t))-I_\Flin(D^n)=t \frac{d I_\Flin}{dt} |_{t_1 \in[0,\delta]} >0$. Similarly, for the opposite case where all the signs of the noise correlations are the same as the signs of $\rhosig_{ij}$, the information will be smaller than the independent case (at least for weak enough correlations). This proves the local ``sign rule" .

Thus, at least for small noise correlations, choosing noise correlations that oppose signal correlations will always be yield higher information values than the case of uncorrelated noise. To prove the ``global" version of this theorem --- that opponent signal and noise correlations always yield better coding than does independent noise --- we will need to establish the convexity of $I_{\Flin}$.  This is done in Theorem~\ref{TH:boundary}. 

Note that, as we will soon prove, $I_{\Flin}$ is a convex function of $t$, and hence $\frac{d I_\Flin}{dt} $ is increasing with $t$. This means that the $\delta$ from our prior argument can be made arbitrarily large, and the same result -- that performance improves when noise correlations are added, so long as they lie along this direction -- will hold, provided that $C^n(\delta)$ is still physically realizable. Thus, the improvement over the independent case is guaranteed globally for any magnitude of noise correlations. 

\end{proof}

Note that the arguments above do not guarantee that the globally optimal noise correlation structure will follow the sign rule.  Indeed, we have seen concrete examples of this in Figs.~\ref{F:2d_sign} and Fig.~\ref{F:3d_boundary}.

\begin{remark}
\label{R:grad_Flin} From Eq.~\eqref{E:dir_prime}, the gradient (steepest uphill direction) of $I_{\Flin}$ evaluated with independent noise $C^n=D^n$ is $(C^{n \prime})_{ij} = -2\frac{ \nabla \mu_i  \cdot \nabla \mu_j} { (D^n)_{ii} (D^n)_{jj}}=-2\frac{ \tnorm{ \nabla \mu_i  }  \tnorm{ \nabla \mu_j }} { (D^n)_{ii} (D^n)_{jj} } \rhosig_{ij}$.
\end{remark}

\begin{remark}
The same result can be shown for $I_\OLE$ and $I_\mutG$, replacing $\nabla \mu$ with $A^0=(C^\mu+D^n)^{-1}L$ and $A^0=(C^\mu+D^n)^{-1}LM^{-\frac{1}{2}}$, respectively, in the definition of $\rhosig$ in Eq.~\eqref{E:rhosig_def}. The gradients are $- 2A^0_i  \cdot A^0_j$ and $-A^0_i  \cdot A^0_j$, respectively, where $A^0_i$ is $i$-th row of $A^0$, and $M=\cov(s,s)- L^T(C^\mu+C^n)^{-1}L$.
\end{remark}

\subsection*{Proof of Theorem \ref{TH:boundary}: optima lie on boundaries}
\label{S:boundary_proof}

We begin by restating Theorem \ref{TH:boundary}, which we then prove first for $I_\OLE$ and then for $I_{\Flin}$ and $I_\mutG$.

\Boundary*

We will show that $I_{\OLE}$ is a convex function of $C^n$ and hence it will either attain its maximum value \emph{only} on the boundary of the allowed region, or it will be uniformly constant. The latter is a trivial case that only happens when $L=0$, as we see below.

\begin{proof}
To show that a function is convex, it is sufficient to show its second derivative along any linear direction is non-negative. For any constant direction $(C^n)^\prime=B$ of changing (off-diagonal entries of) $C^n$, we consider a straight-line perturbation, $C^n(t)=C^n+tB$ parameterized by $t$. Taking the derivative of $I_{\OLE}$ with respect to $t$, 
\beq
\label{e:u_1}
I_{\OLE}^\prime=-\tr(L^T (C^\mu+C^n(t))^{-1}
B
(C^\mu+C^n(t))^{-1} L).
\eeq
We have used that $\frac{dX^{-1}(t)}{dt}=-X^{-1}X^\prime X^{-1}$.
Let $A= (C^\mu+C^n(t))^{-1}L$. Taking another derivative gives
\beq
\label{e:Ipp}
I_{\OLE}^{\prime\prime}=2\tr(A^T B (C^\mu+C^n)^{-1}
B  A )\ge 0.
\eeq
The inequality is because of Lemma \ref{TH:psd_prod} (see below) and $(C^\mu+C^n)^{-1}$ being positive definite. Also, note that $(BA)^T=A^T B$.

For the case when $I_\OLE$ is constant over the region, using Proposition~\ref{TH:convex_0} (below), $B A=0$ for any direction of change $B$. Letting $B_{ij}=\delta_{ip}\delta_{jq}$, $p\neq q$, we see that the $p,q$-th row of $A$ must be 0. This leads to $A=0$ and, since $A = \Gamma^{-1}L$, to $L=0$.  This was the claim in the beginning. In other words, in the case where $I_{\OLE}$ is constant with respect to the noise correlations, the optimal read-out is zero, regardless of the neurons' responses. With the exception of this (trivial) case, the optimal coding performance is obtained when the noise correlation matrix lies on a boundary of the allowed region.
\end{proof}

\begin{lemma}(Linear algebra fact)
\label{TH:psd_prod}
For any positive semidefinite matrix $F$, and any matrix $G$, $G^T F G$ (assuming the dimensions match for matrix multiplications) is positive semidefinite and hence $\tr(G^T F G )\ge 0$. If ``=" is attained, then $FG=0$.
\end{lemma}
\begin{remark}
When $F \succ 0$ i.e. positive definite, $\tr(G^T F G )= 0$ leads to $G=0$ as $F$ is invertible. 
\end{remark}
\begin{proof}
For any vector $z$ (with the same dimension as the number of columns in $G$), $z^T G^T  FG z=(z^T G^T)  F(G z)\ge 0$ since $F\succcurlyeq 0$. Thus, by definition, $G^T  FG \succcurlyeq 0 $, and therefore $\tr ( G^T  FG) \ge 0$. 

For the second part, if $\tr(G^T F G )=0$, all the eigenvalues of $G^T F G $ must be 0 (since none of them can be negative as $G^T  F G\succcurlyeq 0$), hence $G^T F G =0$.  This in fact requires $FG = 0$. To see this, let $U^T \Lambda  U=F$ be an orthogonal diagonalization of $F$. For any vector $z$ as above, $z^T G^T F  Gz=0$. Since the eigenvalues $\Lambda_{ii}$ are non-negative, let $\Lambda^{\frac{1}{2}}$ be the diagonal matrix with the square roots of $\Lambda_{ii}$. We have 
\beq
0 = z^TG^T U^T \Lambda U  Gz=(\Lambda^{\frac{1}{2}} U  Gz)^T (\Lambda^{\frac{1}{2}} U  Gz)=\tnorm{\Lambda^{\frac{1}{2}} U  Gz}^2.
\eeq
Therefore the vector $\Lambda^{\frac{1}{2}} U  Gz=0$ and $FGz=U^T \Lambda^{\frac{1}{2}}\Lambda^{\frac{1}{2}} U  Gz=0$.  
Since $z$ can be any vector, we must have $FG=0$.
\end{proof}

\begin{remark}
Because of the similarities in the formulae for $I_{\OLE}$ and $I_{\Flin}$, the same property can be shown for $I_{\Flin}$. In order for $C^n$ to be invertible, $I_\Flin$ is only defined over the open set of positive definite $C^n$. We therefore assume the closure of the allowed region is contained within this open set $C^n\succ 0$ to state the boundary result.
\end{remark}

A parallel version of Theorem~\ref{TH:boundary} can also be established for $I_\mutG$, as we next show. 

\begin{proof}[Proof of Theorem \ref{TH:boundary} for $I_\mutG$] Again consider the linear parameterization $C^n(t)$ along a direction $B$, as defined above. Let $M=\cov(s,s)- L^T(C^\mu+C^n(t))^{-1}L$. The consistency constraint in Eq.~\eqref{E:mutG_cond} assures $M \succcurlyeq 0$. To keep $I_\mutG$ finite, we further assume $M \succ 0$. Then, the derivative of $I_{\mutG}$ with respect to $t$ is

\beq
I^\prime_{\mutG}=-\frac{1}{2}\frac{\det(M) \tr(M^{-1}M^\prime)}{\det(M)}
=-\frac{1}{2}\tr(M^{-1}M^\prime),
\eeq
where we have used the identity $(\det(M))^\prime=\det(M)\tr(M^{-1}M^\prime)$. The second derivative is thus
\beqrn
I^{\prime\prime}_{\mutG}&=&\frac{1}{2}\tr(M^{-1}M^\prime   M^{-1}M^\prime)
-\frac{1}{2}\tr(M^{-1}M^{\prime\prime})\\
&=&\frac{1}{2}\tr(M^{-1}M^\prime \cdot I  \cdot  M^{-1}M^\prime)\\
&&+\tr(M^{-\frac{1}{2}}L^T\Gamma^{-1}B \cdot \Gamma^{-1} \cdot
B\Gamma^{-1}L
M^{-\frac{1}{2}})\\
&\ge 0&.
\eeqrn
Here $I$ is the identity matrix,
$M^\prime
=L^T\Gamma^{-1}B
\Gamma^{-1}L$,  
$M^{\prime\prime}
=-2L^T \Gamma^{-1}B
\Gamma^{-1}B
\Gamma^{-1}L$ and $\Gamma=C^\mu+C^n$ as defined below Eq.~\eqref{E:A_def}. $M$ being positive definite allows us to split it into its square root $M=M^{\frac{1}{2}}M^{\frac{1}{2}}$.  Moreover, the identity $\tr(PQR) = \tr(QRP)$, for any matrices $P,Q$, and $R$, is used in deriving the last line in the above equation.
 For the last inequality, we apply Lemma \ref{TH:psd_prod} to the two terms with $I$ and $\Gamma^{-1}$ being positive semidefinite.

  We have thus shown that $I_\mutG$ is convex. For the special case that $I_\mutG$ is constant, Proposition~\ref{TH:convex_0} shows $B\Gamma^{-1} L=0$. With the same argument as for $I_{\OLE}$, we observe that, in this (trivial) case $L=0$.
\end{proof}

\subsection*{Proof of Theorem \ref{TH:0noise_cond}: conditions on the noise covariance matrix, under which noise-free coding is possible}
\label{S:0_noise_proof}

We begin by showing that, for a given set of tuning curves, the maximum possible information -- which may or may not be attainable in the presence of noise -- is that which would be achieved if there were no noise in the responses. This is the content of Lemma~\ref{TH:0noise}. Next, we will introduce Lemma~\ref{TH:A_B}, which is a useful linear-algebraic fact that we will use repeatedly in our proofs.

We will then prove Theorem \ref{TH:0noise_cond}, which provides the conditions under which such noise-free performance can be obtained. One direction of the proof of Theorem~\ref{TH:0noise_cond} (sufficiency) is straightforward, while the other direction (necessity) relies on the observation of several conditions  that are equivalent to the one in the theorem. We prove these equalities in Proposition~\ref{TH:convex_0}.

For Theorem \ref{TH:0noise_cond}, we will only consider $I_\OLE$, since $I_\Flin$ and $I_\mutG$ will typically be infinity in the noise-free case ($C^n$ becomes singular). If one takes all instances of infinite information as ``equally optimal,'' a version of Theorem \ref{TH:0noise_cond} can also be obtained; moreover, the condition in Theorem \ref{TH:0noise_cond} becomes a sufficient but not necessary condition for infinite information.

\begin{lemma}[Upper bound by noise-free information]
\label{TH:0noise}

\beq
I_{\OLE} (C^n) \leq I_{\OLE} (0).
\eeq
Here the noise-free information $I_{\OLE} (0)$ refers to that which is obtained when plugging in $0$ at the place of $C^n$ in Eq.~\eqref{E:OLE}. 
\end{lemma}

\begin{proof}
This follows essentially from the consistency between the information quantity and the positive semidefinite ordering of covariance matrices.  First, we write
\beq
I_\OLE(0)-I_\OLE(C^n)=\tr(L^T \left[ (C^\mu)^{-1} -(C^\mu+C^n)^{-1} \right]  L) \; \;.
\eeq
Then, we note the fact that for two positive definite matrices $F,G$, $F \succcurlyeq G $ if and only if $F^{-1} \preccurlyeq G^{-1}$. From this, we have $(C^\mu)^{-1} -(C^\mu+C^n)^{-1} \succcurlyeq 0$.  Finally, applying Lemma~\ref{TH:psd_prod} yields $I_\OLE(0)-I_\OLE(C^n)\ge 0$.

\end{proof}

\begin{lemma}[Useful linear algebra fact]
\label{TH:A_B}
If, for any $F$, $G$, and $M$, $GF^{-1}M=0$, then $(F+G)^{-1}M=F^{-1}M$.
\end{lemma}
\begin{proof}
$(F+G)F^{-1}M=M+GF^{-1}M=M$.
\end{proof}

\begin{proposition}\emph{(Equivalent conditions used in proving the noise-free coding Theorem ~\ref{TH:0noise_cond}).}\\
\label{TH:convex_0} Along a certain direction $(C^n)^\prime=B$, the following conditions are equivalent. 
\beq
(i)~
I_{\OLE}^{\prime\prime}(C^n)=0  \qquad
(ii)~ 
B (C^\mu+C^n)^{-1} L=0 \qquad
(iii)~
 I_{\OLE}(C^n+tB)\equiv I_{\OLE}(C^n).
\eeq
The same also holds for $I_{\Flin}$ and $I_\mutG$.
\end{proposition}

\begin{proof}[Proof for $I_\OLE$]
``$(i) \Leftrightarrow (ii)$":\\
We again consider parametrized deviations from $C^n$, $C^n(t) = C^n + t B$ for some constant matrix B. Let $A_t=(C^\mu+C^n+tB)^{-1}L$, and recall (Eq.~\eqref{e:Ipp}),
\beq
I_{\OLE}^{\prime\prime}(C^n)= 2\tr(A_0^T B (C^\mu+C^n)^{-1} B A_0).
\eeq
Since $C^\mu+C^n$ is positive definite, according to the remark after Lemma~\ref{TH:psd_prod}, we have $(i) \Leftrightarrow (ii)$.

``$(ii) \rightarrow (iii)$)":
If $(ii)$, by Lemma~\ref{TH:A_B}, $A_t=A_0$. We have $I_{\OLE}^{\prime}(C^n+tB)=-\tr(A_t^T B A_t)=-\tr(A_0^T B A_0)=0$, for all $t$ in the allowed region, and hence $(iii)$.

``$(iii) \Rightarrow (i)$": immediate.\\
This concludes the proof for $I_\OLE$.
\end{proof}

\begin{proof}[Proof for $I_\Flin$]
For $I_{\Flin}$, we further assume $C^n \succ 0$ to avoid infinite information. Identical arguments will prove the properties above, where $(ii)$ is replaced by $B(C^n)^{-1}\nabla \mu=0$.
\end{proof}

\begin{proof}[Proof for $I_\mutG$]
For $I_\mutG$, we similarly assume $M \succ 0$ (as defined in the proof of Theorem~\ref{TH:boundary}). Let  $A_t=(C^\mu+C^n+tB)^{-1}L$, then $M^\prime \vert_{t=0}
=A_0^T B A_0$, 
\beqrn
I_{\mutG}^{\prime\prime}(C^n)&=&\frac{1}{2}\tr(M^{-1}M^\prime M^{-1}M^\prime)\\
&&+\tr(M^{-\frac{1}{2}}A_0^T B \cdot \Gamma^{-1} \cdot
B A_0
M^{-\frac{1}{2}}).
\eeqrn
It is easy to see $(ii) \Rightarrow (i)$. When $(i)$ holds, using Lemma~\ref{TH:psd_prod}, each of the two terms must be 0. In particular, as we discussed in the proof of Theorem \ref{TH:boundary} for $I_\mutG$ (above), each of the terms is non-negative. Thus, if their sum is $0$, then each term must individually be $0$. According to the remark after Lemma~\ref{TH:psd_prod}, the second term being 0 indicates that $B A_0 M^{-\frac{1}{2}}=0$ or $B A_0 =0$, which is $(ii)$.

If $(ii)$ holds, by Lemma~\ref{TH:A_B}, we have $A_t=A_0$. We have $I_{\mutG}^{\prime}(C^n+tB)=-\frac{1}{2}\tr(M^{-1}A_t^T B A_t)=-\frac{1}{2}\tr(M^{-1}A_0^T B A_0)=0$, for all $t$ in the allowed region, and hence $(iii)$. Similarly $(iii) \Rightarrow (i)$. This proves the property for $I_\mutG$.
\end{proof}

\ZeroNoiseCond*

\begin{proof}
If $C^n (C^{\mu})^{-1}L=0$, then Lemma~\ref{TH:A_B} implies that  $(C^{\mu} + C^n)^{-1}L= (C^{\mu})^{-1}L $, which means that $I_{\OLE}(C^n) = I_{\OLE}(0)$, using the definition in Eq.~\eqref{E:OLE}.

For the other direction of the theorem, consider a function of $t\in[0,1]$, $I_{\OLE}(tC^n)=\tr(L^T(C^\mu+tC^n)^{-1}L)$, whose values at the endpoints are equal, according to saturation of the information bound. The mean value theorem assures that there exists a $t_1\in [0,1]$ such that 
\beq
I_{\OLE}^\prime(t_1 C^n )=-\tr(L^T(C^\mu+t_1 C^n )^{-1} C^n (C^\mu+t_1 C^n)^{-1}L)=0.
\eeq
Since $ C^n$ is positive semidefinite, according to Lemma~\ref{TH:psd_prod}, $C^n (C^\mu+t_1 C^n)^{-1}L=0$. Now using Lemma~\ref{TH:A_B}, we have that $C^n (C^\mu)^{-1}L=0$, and the readout vector $A=(C^\mu+C^n)^{-1}L=(C^\mu)^{-1}L$.
\end{proof}

\subsection*{Proof of Theorem \ref{TH:0noise_cond_diag}: conditions on tuning curves and variance, under which noise-free coding performance is possible}
\label{S:0_noise_cond_proof}

Next, we will restate, and then prove, Theorem \ref{TH:0noise_cond_diag}. The proof will require using geometric ideas in Lemma~\ref{TH:q_to_B}, which we will state and prove below.

\ZeroNoiseCondDiag*

The proof is based on the condition in Theorem~\ref{TH:0noise_cond}. After taking several invertible transforms of the equation, the problem of finding a noise-canceling $C^n$ is transformed to that of finding a set of $N$ vectors, whose length are specified by $q_i$, that sum to zero (the vectors form a closed loop when connected consecutively). This allows us to take a geometrical point of view, in which inequality Eq.~\eqref{E:max_cond} becomes the triangle inequality.  This will prove the ``necessary" part of the Theorem.  Lemma~\ref{TH:q_to_B} shows the opposite direction, by inductively constructing the set of vectors that sum to zero. 

This procedure will yield one ``particular" $C^n$ with the noise-canceling property. Very much like finding all general solutions of an ODE, we then add to our particular solution an arbitrary homogeneous solution, which belongs to a vector space of dimension $\frac{N(N-3)}{2}$. In order for our perturbed solution, at least for small enough perturbations, to still be positive semidefinite, the particular $C^n$ we start with must be generic.  In other words, it must satisfy a rank condition, which is guaranteed by the construction in Lemma~\ref{TH:q_to_B}. We can then conclude that the set of all noise canceling $C^n$ forms a linear segment with the dimension of the space of homogeneous solutions.

Finally, special treatments are given for the cases of ``$=$" in Eq.~\eqref{E:max_cond}, as well as cases where some $q_i\text{'s}$ are 0.

\begin{proof} 
To establish the necessity direction of the Theorem, first let $D$ be a diagonal matrix with $D_{ii}=A_i$ or $A=De$, where vector $e=(1,\cdots,1)^T$.  Note that
\beq
\label{E:p1}
C^nA=0 \Rightarrow DC^nDe=0 \; \;.
\eeq 
Let $DC^nD=\tilde{C^n}$, a positive semidefinite matrix with diagonal $\{q_i^2\}$.

$\tilde{C^n}$ can be diagonalized by an orthogonal matrix $U$, $\tilde{C^n}=U^T\Lambda U$. Without loss of generality, further assume that the first $k$ diagonal elements of $\Lambda$ are positive, with the rest being 0, where $k=\rank(C^n)$. Let $\hat{\Lambda}$ be the first $k$ block of $\Lambda$,  and $\hat{U}$ be the first $k$ rows of $U$.  Then we have
\beq
\label{E:p2}
U^T\Lambda U e=0 \Rightarrow \Lambda U e=0 \Rightarrow \hat{\Lambda}^{\frac{1}{2}}\hat{U} e=0 \; \; .
\eeq 
Let $B=\hat{\Lambda}^{\frac{1}{2}}\hat{U}$, a $k\times N$ matrix, and $B_{i}$ be the $i$-th column. As $\tilde{C^n}=B^T B$, the 2-norm of vector $B_i$ is $q_i$. Let $q_j$ be the maximum of $\{q_i\}$,
\beq
\label{E:p3}
Be=0 \Rightarrow \sum_{i} B_i=0 \Rightarrow -B_j=\sum_{i\neq j} B_i
\Rightarrow \tnorm{ B_j} \leq \sum_{i\neq j} \tnorm{ B_i }.
\eeq
This concludes the necessary direction of our proof.

To establish sufficiency, we first focus on the case of ``$<$" and all $A_i\neq 0$. We will construct a generic $C^n$  that has rank $N-1$, satisfying $C^nA=0$. We will basically reverse the direction of arguments in Eq.~(\ref{E:p1}-\ref{E:p3}). We will later deal with the``$=$" case, and the case of $A_i = 0$ for some $i$.

\begin{lemma}
\label{TH:q_to_B}
Let $e_i$, $i=1,\cdots,N-1$ be an orthonormal basis of $\mathbb{R}^{N-1}$. Given a set of positive $\{q_i\}_{i=1}^{N}$ satisfying ``$<$" in Eq.~\eqref{E:max_cond}, there exist $N$ vectors $\{B_i\}$, such that $\sum_{i} B_i=0$, $\tnorm{ B_i}=q_i$ and the spanned linear subspace $\vspan\{B_i\}_{i=1}^N=\vspan\{e_i\}_{i=1}^{N-1}$. 
\end{lemma}

\begin{proof}
We prove this by induction. $N$ has to be at least 3 for the inequality to hold. For $N=3$, this is the case of a triangle. There is a (unique) triangle $X_1 X_2 X_3$, for which the length of the three sides $X_1 X_2$, $X_2 X_3$, $X_1 X_3$ are $q_3,q_1,q_2$ respectively. The altitude from $X_3$ intersects the line of $X_1 X_2$ at $O$. Let $O$ be the origin of the coordinate system, with $X_1 X_2$ being the x-axis and aligned with $e_2$, and the altitude $O X_3$ being the y-axis aligned with $e_1$. From such a picture, it is easy to verify the following: $B_3=-(\abs{X_1 O}+p\abs{X_2 O}) e_2$, $B_1=\abs{X_1 O} e_2+\abs{O X_3} e_1$, $B_2=p\abs{X_2 O} e_2-\abs{O X_3} e_1$ satisfies the lemma, where $p=1$ if $O$ lies within $X_1X_2$ and $p=-1$ otherwise. 

For the case of $N\ge 4$, assume  that $q_N$ is the largest of the $q\text{'s}$. Because of the inequality, there will always exist some non-negative real number $q$ (not necessarily one of the $q_i\text{'s}$) such that
\beq
\max\{q_N-q_{N-1},q_1,\cdots,q_{N-2} \}<q<\min\{ \sum_{i=1}^{N-2} q_i, q_N+q_{N-1}\}.
\eeq
We can verify that the set $\{q_1,\cdots,q_{N-2}, q\}$ satisfies the inequality as well. By the assumption of induction, there exist vectors $\{B_1,\cdots, B_{N-2}, B\}$ that span the space of $\{e_1,\cdots,e_{N-2}\}$, such that $\tnorm{ B_i}=q_i$ and $\tnorm{ B}=q$.

Note the choice of $q$ also guarantees that $\{q_{N},q_{N-1},q\}$ can be the edge lengths of a triangle. Applying the result at $N=3$, the three sides $X_1 X_2$, $X_2 X_3$, $X_1 X_3$ correspond to $q, q_{N},q_{N-1}$ respectively. Let $B_{N-1}=-\frac{|X_1 O|}{\abs{X_1 O}+p\abs{X_2 O}} B+|O X_3| e_{N-1}$, $B_{N}=-p\frac{|X_2 O|}{\abs{X_1 O}+p\abs{X_2 O} } B-|O X_3| e_{N-1}$. It is easy to verify that these $\{B_{i}\}_{i=1}^{N}$ satisfy the lemma. 

\end{proof}

Using the lemma, we have a set of $B_i$. Stacking them as column vectors gives a matrix $B$; moreover, $Be=0$. Let $\tilde{C^n}=B^TB$, which is positive semidefinite with diagonals $\{q_i^2\}$. It is easy to show that $\rank(\tilde{C^n})=\rank(B)=N-1$, by comparing the null spaces of the matrices. Let $C^n=D^{-1}\tilde{C^n}D^{-1}$, where $D$ is defined as above. Then $C^nA=D^{-1}\tilde{C^n} D^{-1}A=D^{-1}\tilde{C^n} e=0$. 

Now consider the case where there are zeros in $A_i$.  Assume that the first $k$ entries contain all of the the non-zero values. We apply the construction above for the first $k$ dimensions, and get a $k\times k$ matrix such that $\hat{C^n}\hat{A}=0$, $\rank(\hat{C^n})=k-1$, where $\hat{A}$ is part of $A$ with the first $k$ elements. The following block diagonal matrix 
\beq
C^n=
\left(
\barr{cc}
\hat{C^n} & 0\\
0 & \tilde{D^n}
\earr
\right),\quad   \text{where} \quad \tilde{D^n}=\diag\{C^n_{k+1,k+1},\cdots C^n_{N,N}\},
\eeq
satisfies $C^nA=0$ and $\rank(C^n)=N-1$.

We have shown that for the ``$<$" case in the theorem, there is always a noise canceling $C^n$. Consider the direction $(C^n)^\prime$, in which off-diagonal elements of $C^n$ vary, while keeping $(C^n+ (C^n)^\prime)A=0$ (temporarily ignoring the positive semidefinite constraint). The set of all such $(C^n)^\prime$  form a linear subspace $M$ of $\mathbb{R}^{\frac{N(N-1)}{2}}$, determined by the linear system $(C^n)^\prime A=0$. Since there are $N$ equations, the dimension of $M$ is at least $\frac{N(N-1)}{2}-N=\frac{N(N-3)}{2}$.

In the ``$<$" case, there must be at least 3 non-zero $A_i\text{'s}$ in order for the triangle inequality to be satisfied in Eq.~\ref{E:max_cond}.  We will choose these three $A_i\text{'s}$ to be $A_1,A_2,A_N \neq 0$. Consider a block of the coefficient matrix associated with the system $(C^n)^\prime A=0$ (note that the entries of $(C^n)^\prime$ are considered to be unknown variables), that are columns corresponding to variables $(C^n)^\prime_{12}, \cdots, (C^n)^\prime_{1N}, (C^n)^\prime_{2N}$
\beq
\left(
\barr{cccccc}
A_2 & A_3 & \cdots & \cdots & A_N &  0\\
A_1 & & & & &A_N\\
& A_1& & \bigzero & &0\\
& & \ddots& & & \vdots\\
& \bigzero & &\ddots & &0 \\
& & & & A_1 & A_2
\earr
\right).
\eeq

Performing Gaussian elimination on the columns of this matrix, we obtain the following matrix, which will have the same rank.

\beq
\left(
\barr{cccccc}
A_2 & A_3 & \cdots & \cdots & A_N &  -\frac{2 A_2 A_N}{A_1} \\
A_1 & & & & &0\\
& A_1& & \bigzero & &\\
& & \ddots& & & \vdots\\
& \bigzero & &\ddots & &\\
& & & & A_1 &0
\earr
\right)
\eeq
This matrix -- which determines the number of constraints that must be satisfied in order for $(C^n)^\prime A=0$ -- has rank $N$, and hence $\dim(M)$ is exactly $\frac{N(N-1)}{2} - N =  \frac{N(N-3)}{2}$.

For any direction in $M$, we can always perturb the generic $C^n$ we found above by some finite amount $\epsilon$, and still have $C^n+\epsilon (C^n)^\prime $ be positive semidefinite. Let $\lambda_{\min}$ be the smallest non-zero eigenvalue of $C^n$. Take any $|\epsilon|<\frac{1}{2} \frac{\lambda_{\min}}{\tnorm{ (C^n)^\prime}}$. For any vector $z$, let $z=z_0+z_1$ be an orthogonal decomposition where $z_0=\frac{1}{\tnorm{ A}^2}AA^T z$ is the projection along the direction of $A$. Then
\beq
z^T(C^n+\epsilon (C^n)^\prime)z =z_1^T C^n z_1+\epsilon z_1^T (C^n)^\prime z_1 \ge \lambda_{\min} \tnorm{ z_1}^2 -\abs{\epsilon} \tnorm{ z_1}^2 \tnorm{ (C^n)^\prime} \ge 0.
\eeq

This shows that the $C^n+\epsilon (C^n)^\prime$ are positive semidefinite and they form a set of dimension as $M$. We can always take the admissible $\epsilon$ values to their extremes, and the resulting matrices are all the possible noise canceling $C^n$. For any $\tilde{C^n}A=0$, $(C^n-\tilde{C^n})A=0$, and $C^n-\tilde{C^n}$ must be in $M$. Note that the sets of positive semidefinite $C^n$ (spectrahedra) are convex.  As a consequence, any point along the segment $C^n+t(\tilde{C^n}-C^n)$ will be positive semidefinite. This shows we must have encompassed $\tilde{C^n}$ when considering the largest possible perturbations of $C^n$, in any direction $(C^n)^\prime \in M$. Moreover, we note that the set of all noise-canceling $C^n$ is convex: if $C^n_{i}A=0$, $i\in\{1,2\}$, $(\lambda C^n_1+(1-\lambda)C^n_2)A=0$ for any $\lambda \in [0,1]$ and $(\lambda C^n_1+(1-\lambda)C^n_2)$ is positive semidefinite, with the diagonal matching $C^n_{ii}$. 

Thus, we have proved the claim about the dimension and convexity of the set of optimal correlations for the case of ``$<$" in Eq.~\eqref{E:max_cond}.

Finally, for the special case of ``$=$" in Eq.~\eqref{E:max_cond}, again first consider the case where all $A_i\neq 0$. As before, solving $C^nA=0$ is equivalent to solving $\tilde{C^n}e=0$ and there is an one to one correspondence between the two. Revisiting Eq.~\eqref{E:p3} in the proof above, the equality condition in the triangle inequality implies that $\{B_i, i=1,\cdots,N-1\}$ all point along the same direction, and that $B_N$ is in the opposite direction, in order to cancel their sum. This fully determines $\tilde{C^n}=D^q C^0 D^q$, where $D^q=\diag\{q_1,\cdots q_N\}$, and

\beq
\label{E:C^0}
C^0=
\left(
\barr{ccccc}
1 & \cdots & \cdots& 1& -1\\
\vdots & \ddots &  &\vdots &\vdots\\
\vdots &   & \ddots& 1&-1 \\
1 & \cdots & 1& 1 &-1\\
-1 & \cdots &-1 &-1 &1
\earr
\right).
\eeq
It is easy to verify that $\tilde{C^n}e=0$, and hence there is a unique noise canceling $C^n$.

For the case when there are $N_0$ 0's among the $\{A_i\}$, assume that the first $N-N_0$ coordinates are non-zero, so that $A=(\hat{A},0,\cdots,0)^T$. Next, we write $C^n A=0$ in block matrix form, with blocks of dimension $N-N_0$ and $N_0$:
\beq
C^n A=
\left(
\barr{cc}
\hat{C^n} & E^T\\
E &F
\earr
\right)
\left(
\barr{c}
\hat{A} \\
0
\earr
\right)
=\left(
\barr{c}
\hat{C^n}\hat{A} \\
E \hat{A}
\earr
\right)
=0.
\eeq
Applying the previous argument from the $A_i\neq 0$ case, there is a unique $\hat{C^n}$.  Moreover, note that $\rank(\hat{C^n})=1$, following from the fact that $C^0$ in Eq.~\eqref{E:C^0} has rank 1. Let $\hat{C^n}=U^T \Lambda U$ be the orthogonal diagonalization and $\Lambda_{N-N_0,N-N_0}=\lambda\neq 0$. Let $I_{N_0}$ be the identity matrix of dimension $N_0$. Then we can take an orthogonal transform:

\beqrn
&&
\left(
\barr{cc}
U^T & 0\\
0 & I_{N_0}
\earr
\right)
\left(
\barr{cc}
\hat{C^n} & E^T\\
E &F
\earr
\right)
\left(
\barr{cc}
U & 0\\
0 & I_{N_0}
\earr
\right)
\left(
\barr{cc}
U^T & 0\\
0 & I_{N_0}
\earr
\right)
\left(
\barr{c}
\hat{A} \\
0
\earr
\right)\\
&=&
\left(
\barr{cc}
\Lambda & U^T E^T\\
EU & F
\earr
\right)
\left(
\barr{c}
U^T\hat{A} \\
0
\earr
\right).
\eeqrn
With the notation $U^T \hat{A}=\hat{A}^\prime$, the original problem $C^n A=0$ is therefore equivalent to finding all $E^\prime$ and $F$ such that,
\beq
\label{E:A'}
\left(
\barr{cc}
\Lambda & E^{\prime T} \\
E^\prime & F
\earr
\right)
\left(
\barr{c}
\hat{A}^\prime \\
0
\earr
\right)
=0,
\eeq
while keeping the matrix in this equation positive semidefinite.

 For any positive semidefinite matrix $X$, it is easy to show that $X_{ii}X_{jj} \ge X_{ij}^2$ by considering the principle minor with indices $i,j$, which must be non-negative. Note that since $\Lambda$ has only one non-zero diagonal entry, this forces the first $N-N_0-1$ columns of $E^\prime$ to be entirely 0. So we can rewrite the block matrix by dimension $N-N_0-1$ and $N_0+1$ as
\beq
\label{E:=block}
\left(
\barr{cc}
\Lambda & E^{\prime T} \\
E^\prime & F
\earr
\right)
=
\left(
\barr{ccc}
0 & 0& 0\\
0 & \lambda & e^{\prime T} \\
0& e^\prime & F
\earr
\right),
\eeq
where $e^\prime$ is the $(N-N_0)$-th column of $E^\prime$.
Since $\Lambda \hat{A}^\prime= \lambda(\hat{A}^\prime)_{N-N_0}=0$, we have $(\hat{A}^\prime)_{N-N_0}=0$. It can be verified that, as long as the block structure of Eq.~\eqref{E:=block} is satisfied, Eq.~\eqref{E:A'} is always true. The positive semidefinite constraint becomes the constraint that the lower block be positive semidefinite; in turn, this corresponds to a spectrahedron (and hence a convex set) of dimension $\frac{N_0(N_0+1)}{2}$. Note that this dimensionality and convexity will be preserved when we undo the invertible linear transforms performed in prior steps to obtain the noise-canceling $C^n\text{'s}$.
\end{proof}

\subsection*{Proof of Theorem \ref{TH:0noise_prob}: probability that noise-free coding is possible}
\label{S:prob_0_noise_proof}

In this subsection, we will restate, and then prove, Theorem \ref{TH:0noise_prob}.

\ZeroNoiseProb*

\begin{proof}

We will use the following fact to establish a lower bound for the probability of the event in the theorem (below, we denote this event as $C$). 
\beq
\label{E:PAB}
P(A\cap B)\geq P(A)+P(B)-1
\eeq
We choose the two events $A$ and $B$ as $A=\frac{1}{N}\sum_i q_i > \frac{2}{3} \EVb{X}$ and $B=\max\{ q_i\} \leq \frac{N}{4} \EVb{X} $. Note that $A\cap B$ implies C,
\beq
\max\{ q_i\} \leq \frac{N}{4} \EVb{X} <  \frac{N}{3} \EVb{X}\leq\frac{1}{2}\sum_i q_i,
\eeq
the event in concern. We will then show that, for large populations, $P(A) \to 1$ and $P(B) \to 1$, and thus $P(C)\ge P(A\cap B) \to 1$.

For $A$, by the law of large numbers, the average should converge to the expectation (which is a positive number), hence
\beq
\lim_{N\rightarrow \infty}P\left(\frac{1}{N}\sum_i q_i > \frac{2}{3} \EVb{X} \right)=1.
\eeq

We next consider event B.  Let the cumulative distribution function of $X$ be $F(x)$. Then cumulative distribution function for $\max \{q_i\}$ is $F^N(x)$ by the assumption that these variables are drawn i.i.d.  It follows that 
\beqrn
&&P\left(\max\{ q_i\} > \frac{N}{4} \EVb{X} \right)=\int_{\frac{N}{4} \EVb{X}}^{\infty} N F^{N-1}(x) dF(x)\\
&\leq& \int_{\frac{N}{4} \EVb{X}}^{\infty} \frac{4x}{\EVb{X}}F^{N-1}(x) dF(x)
\leq  \frac{4}{\EVb{X}} \int_{\frac{N}{4} \EVb{X}}^{\infty}  x dF(x).
\eeqrn
Here, the first inequality is obtained via the lower bound of $x$ over the interval of integration, and the second uses the fact $F(x)\leq 1$.

As $N\rightarrow \infty$, the last integral converges to 0 because of the fact that $\EVb{X}<\infty$, together with the Lebesgue dominated convergence theorem. Hence $P\left(\max\{ q_i\} \leq \frac{N}{4} \EVb{X} \right) \rightarrow 1$ as $N\rightarrow \infty$.

Combining the limits of $A$ and $B$ using Eq.~\eqref{E:PAB}, together with the fact $C\supseteq A\cap B$, we conclude that $P(C)$ must approach 1 as $N\rightarrow \infty$.
\end{proof}

\subsection*{Proof of Proposition~\ref{TH:cond_Cn}: sensitivity to perturbations}
\label{S:robust_proof}

Here, we will prove Proposition~\ref{TH:cond_Cn}, which puts bounds on the condition numbers that define the sensitivity of our coding metrics to perturbations in noise correlations or the tuning curves.  For our proof, we will require three different lemmas.  We state and prove these, before moving on to Proposition~\ref{TH:cond_Cn}.

Here, we will first consider the condition number for the case of a scalar stimulus $s$, when $L$ is a vector. In the proof of the proposition, we show how to extend the results to the case of multivariate $s$. As we mentioned in Section~\namerefquote{S:robust}, the same proof works for $I_\Flin$ as well as $I_\OLE$.

\begin{lemma}
\label{TH:1/(I+A)}
For any submultiplicative matrix norm $\norm{ \cdot }$ and $\norm{A} \leq 1/2$,
\beq
\norm{ (I-A)^{-1} }\leq 2.
\eeq
\end{lemma}
\begin{proof}
Since $\norm{ A } \leq 1/2$, $(I-A)^{-1}$ exists and 
\beq
\norm{ (I-A)^{-1} }= \norm{ \sum_{n=0}^{\infty} A^n}
\leq  \sum_{n=0}^{\infty} \norm{ A }^n = (1-\norm{ A})^{-1}
\leq 2.
\eeq
\end{proof}

\begin{lemma}
\label{TH:dl} For any positive definite matrix $A$, vectors $l$ and $a$ such that $\tnorm{ a } \tnorm{ A}^{\frac{1}{2}} \tnorm{ A^{-1}}^{\frac{1}{2}} \leq \tnorm{ l }$,
\beq
\frac{| (l+a)^T A^{-1} (l+a) - l^T A^{-1} l | }
{|  l^T A^{-1} l  |}
\leq 3 \tnorm{A^{-1}}^{\frac{1}{2}}  \tnorm{A}^{\frac{1}{2}}    \frac{\tnorm{ a}} {\tnorm{ l} }. 
\eeq
\end{lemma}

\begin{proof}
\beqrn
&& | (l+a)^T A^{-1} (l+a) - l^T A^{-1} l |  \leq 
2| a^T A^{-1} l | +|a^T A^{-1} a | \\
&=& 2 | a^T A^{-\frac{1}{2}}  A^{-\frac{1}{2}}   l |  
+ |a^T A^{-1} a |  \\
&\leq& 2 \tnorm{ a} \tnorm{ A^{-\frac{1}{2}} } \tnorm{ A^{-\frac{1}{2}} l }
+ \tnorm{ a}^2 \tnorm{ A^{-1} }\\
& = & 2\frac{ \tnorm{ a}}{\tnorm{ l }}  \tnorm{ A^{-\frac{1}{2}} } \tnorm{ A^{-\frac{1}{2}} l }^2 \frac{\tnorm{ l }}{\tnorm{ A^{-\frac{1}{2}} l }}
+ \left(\frac{ \tnorm{ a}}{\tnorm{ l }}  \frac{\tnorm{ l }}{\tnorm{ A^{-\frac{1}{2}} l }} 
 \right)^2     \tnorm{ A^{-\frac{1}{2}} l }^2      \tnorm{ A^{-1} }\\
 & \leq & 2\frac{ \tnorm{ a}}{\tnorm{ l }}  \tnorm{ A^{-\frac{1}{2}} } \tnorm{ A^{-\frac{1}{2}} l }^2 
  \tnorm{ A^{\frac{1}{2}} }
  +\left(\frac{ \tnorm{ a}}{\tnorm{ l }}   \tnorm{ A^{\frac{1}{2}} }
 \right)^2      \tnorm{ A^{-\frac{1}{2}} l }^2     \tnorm{ A^{-1} }\\
 & \leq & 3 \tnorm{ A^{-1} }^{\frac{1}{2}}  \tnorm{ A }^{\frac{1}{2}}    \frac{\tnorm{ a }} {\tnorm{ l }} |  l^T A^{-1} l  | .
   \eeqrn
   Here, we have used  $\tnorm{  A^{-\frac{1}{2}} }
=\tnorm{  A^{-1} }^{\frac{1}{2}}$, $\tnorm{ A^{-\frac{1}{2}} l }^2= |  l^T A^{-1} l  |$, and the assumed condition in the last line.
\end{proof}

\begin{lemma}
\label{TH:d1/A}
For any positive definite matrix $A$, vector $l$ and matrix $B$ where $\tnorm{A^{-1} } \tnorm{B } \leq 1/2 $,
\beq
\frac{\abs{l^T A^{-1} l - l^T (A+B)^{-1} l } }
{\abs{  l^T A^{-1} l  }}
\leq 2 \tnorm{ A^{-1}} \tnorm{B}. 
\eeq
\end{lemma}

\begin{proof}
\beqrn
&&\abs{ l^T A^{-1} l - l^T (A+B)^{-1} l } 
=\abs{l^T (A+B)^{-1} B A^{-1} l }\\
&=&\abs{ l^T A^{-\frac{1}{2}} 
(I+A^{-\frac{1}{2}}BA^{-\frac{1}{2}})^{-1} 
A^{-\frac{1}{2}}B A^{-\frac{1}{2}} A^{-\frac{1}{2}} l }\\
&\leq&
\tnorm{l^T A^{-\frac{1}{2}} }
\tnorm{(I+A^{-\frac{1}{2}}BA^{-\frac{1}{2}})^{-1} }
\tnorm{A^{-\frac{1}{2}} }     \tnorm{B }     \tnorm{A^{-\frac{1}{2}} }    \tnorm{A^{-\frac{1}{2}} l  }\\
&=&
( l^T A^{-1} l )
\tnorm{A^{-1} }   \tnorm{B }
\tnorm{ (I+A^{-\frac{1}{2}}BA^{-\frac{1}{2}})^{-1} }\\
&\leq&
2(  l^T A^{-1} l  )
\tnorm{A^{-1} }   \tnorm{B }.
\eeqrn
Here we have used $\tnorm{A^{-\frac{1}{2}} }=\tnorm{A^{-1} }^{\frac{1}{2}}$. As $\tnorm{A^{-\frac{1}{2}}BA^{-\frac{1}{2}}}  \leq   \tnorm{A^{-1} }  \tnorm{B}\leq 1/2$, we apply Lemma \ref{TH:1/(I+A)} is applied to obtain the last line.
\end{proof}

\CondCn*

\begin{proof}

Note that
\beq
I_\Flin=\tr(\nabla \mu^T (C^n)^{-1} \nabla \mu)
=\sum_{i=1}^K e_i^T \nabla \mu^T (C^n)^{-1} \nabla \mu e_i,
\eeq
where $e_i=(0,\cdots,1,\cdots,0)^T$ is the $i$-th unit vector ($K\times 1$). Since the bound in Lemma~\ref{TH:d1/A} does not depend on $l$, we apply the Lemma for $l=\nabla \mu e_i=(\nabla \mu)_{\cdot,i}$ and each $i$ respectively.  For any perturbation $B$ satisfying $\tnorm{(C^n)^{-1}}  \tnorm{B} \leq 1/2 $, we have
\beqrn
\frac{| I_\Flin(C^n)-I_\Flin(C^n+B)|} { I_\Flin(C^n)} &\leq& \frac{1}{ I_\Flin(C^n)} \sum_{i=1}^K 2\kappa_2 \frac{\tnorm{ B}}{\tnorm{ C^n}}  e_i^T \nabla \mu^T (C^n)^{-1} \nabla \mu e_i \\
&\leq& 2\kappa_2 \frac{\tnorm{ B}}{\tnorm{ C^n}}.
\eeqrn
Here $\kappa_2=\tnorm{ (C^n)^{-1}}\cdot  \tnorm{ C^n}$.  We then note that for positive semidefinite matrices $\tnorm{ C^n}=\lambda_{\max}$, $\tnorm{ (C^n)^{-1}}= \lambda_{\min}^{-1}$, where $\lambda_{\min}$ and $\lambda_{\max}$ are the smallest and largest eigenvalues of $C^n$. This proves the bound on the condition number for perturbing $C^n$.

Similarly, for a perturbation of $\nabla \mu$ with $\nu$, $\tnorm{ \nu } \tnorm{ A}^{\frac{1}{2}} \tnorm{ A^{-1}}^{\frac{1}{2}} \leq \min_i \tnorm{ \nabla \mu e_i }$. This guarantees that 
\beq
\tnorm{ \nu e_i } \tnorm{ A}^{\frac{1}{2}} \tnorm{ A^{-1}}^{\frac{1}{2}} 
\leq 
 \tnorm{ \nu  }  \tnorm{ e_i } \tnorm{ A}^{\frac{1}{2}} \tnorm{ A^{-1}}^{\frac{1}{2}}
 \leq  \tnorm{ \nabla \mu e_i }. 
\eeq
Applying Lemma~\ref{TH:dl} for each $\nu e_i$ and $\nabla \mu e_i$, we have
\beqrn
&&\frac{| I_\Flin(\nabla \mu)-I_\Flin(\nabla \mu+\nu)|} { I_\Flin(\nabla \mu)} \leq \frac{1}{I_\Flin(\nabla \mu)} \sum_{i=1}^K 3\sqrt{\kappa_2} \frac{\tnorm{ \nu e_i }}{\tnorm{ \nabla \mu e_i }}  e_i^T \nabla \mu^T (C^n)^{-1} \nabla \mu e_i \\
&\leq& 3\sqrt{\kappa_2} \frac{\tnorm{ \nu }}{\min_i \tnorm{ \nabla \mu e_i }} 
= 3\sqrt{\kappa_2}
\frac{\tnorm{ \nabla \mu }}{\min_i \tnorm{ \nabla \mu e_i }} 
 \frac{\tnorm{ \nu }}{\tnorm{ \nabla \mu }} 
 \leq 3\sqrt{\kappa_2}
\frac{\norm{ \nabla \mu }_F}{\min_i \tnorm{ \nabla \mu e_i }} 
 \frac{\tnorm{ \nu }}{\tnorm{ \nabla \mu }} \\
 &\leq&  3\sqrt{\kappa_2}
\frac{\sqrt{K}\max_i\tnorm{ \nabla \mu e_i }}{\min_i \tnorm{ \nabla \mu e_i }} 
 \frac{\tnorm{ \nu }}{\tnorm{ \nabla \mu }}  .
\eeqrn
Here $\norm{\cdot}_F$ is the Frobenius norm and we have used the fact for any matrix $G$, $\norm{G}_2\leq \norm{G}_F$. The last inequality follows  from the definition of $\norm{\cdot}_F$.
\end{proof}

\subsection*{Details for numerical examples and simulations}
\label{S:numerical_details}

Here, we describe the parameters of our numerical models, and the numerical methods we used.

\subsubsection*{Parameters for Fig.~\ref{F:tuning}, Fig.~\ref{F:2d_sign} and Fig.~\ref{F:3d_boundary}}
\label{S:numerics}
All parameters we use are dimensionless, unless stated otherwise.

In Fig.~\ref{F:tuning}, the mean response for the three neurons under stimulus 1 (red) and 2 (blue) is $\mu_1$ and $\mu_2$ respectively:
\beq
\mu_1=
\left(
\barr{c}
1.48 \\
1.16 \\
1.16
\earr
\right),
\quad
\mu_2=
\left(
\barr{c}
0.52 \\
0.84 \\
0.84
\earr
\right).
\eeq
For each case of correlation structure (i.e., for each row in Figure~\ref{F:tuning}), the noise covariance matrix is the same for the two stimuli, and all neuron variances $C^n_{ii}=1$.  In detail:  
\beq
C^n_{\text{A}}=
\left(
\barr{ccc}
1 &  & \\
 & 1 &  \\
  &  & 1
\earr
\right),
\quad
C^n_{\text{B}}=
\left(
\barr{ccc}
1 & 0.3381 & 0.3381 \\
0.3381 & 1 & 0.1127  \\
0.3381  & 0.1127 & 1
\earr
\right),
\eeq
\beq
C^n_{\text{C}}=
\left(
\barr{ccc}
1 & -0.1833 & -0.1833 \\
 -0.1833 & 1 & -0.7333 \\
-0.1833  & -0.7333 & 1
\earr
\right),
\quad
C^n_{\text{D}}=
\left(
\barr{ccc}
1 & -0.405 & 0.675\\
-0.405 & 1 & 0.225 \\
 0.675 & 0.225 & 1
\earr
\right).
\eeq

The confidence circles and spheres are calculated based on a Gaussian assumption for the response distributions.

In Fig.~\ref{F:2d_sign}, the noise variances $C^n_{ii}$ are all set to 1.  Additionally, 
\beq
C^\mu=
\left(
\barr{ccc}
0.8 &  & \\
 & 0.8 &  \\
  &  & 0.8
\earr
\right),
\quad
L=
\left(
\barr{c}
0.0310 \\
0.4012 \\
0.0406
\earr
\right).
\eeq

In Fig.~\ref{F:3d_boundary}, the noise variances $C^n_{ii}$ are all set to 1. In panel {\bf A} 
\beq
C^\mu=
\left(
\barr{ccc}
1 & 0.3 & 0.2\\
0.3 & 1 & -0.1 \\
0.2  & -0.1 & 1
\earr
\right),
\quad
L=
\left(
\barr{c}
1 \\
1 \\
1
\earr
\right).
\eeq
For panel {\bf B} 
\beq
C^\mu=
\left(
\barr{ccc}
1 &  & \\
 & 1 &  \\
  &  & 1
\earr
\right),
\quad
L=
\left(
\barr{c}
1 \\
0 \\
0
\earr
\right).
\eeq

\subsubsection*{Heterogeneous tuning curves}

For the results in Section~\namerefquote{S:heter_eg}, we use the same model and parameters as in \cite{Ecker:2011bx} to set up a heterogeneous population with tuning curves of random amplitude and width. For completeness, we include the details of this setup as follows:

The shape of each tuning curve (specifying firing rates) is modeled by a von Mises distribution.  This an analog of the Gaussian distribution over the unit circle:
\beq
r_i(\theta)=\alpha_i+\beta_i \exp(\gamma_i [\cos(\theta-\phi_i)-1]).
\eeq
The parameters $\beta_i,\gamma_i, \; \text{and} \; \phi_i$ respectively control the magnitude, width and preferred direction for each neuron. We set $\phi_i$ to be equally spaced along $[0, 2\pi]$ and $\alpha_i=1$. The $\beta_i$ are independently chosen from a $\chi$-square distribution with 3 degrees of freedom, scaled to a mean of 19. $\gamma_i$ is similarly drawn from a $\log$-normal distribution with parameters giving mean 2 and standard deviation 2 (for the underlying normal distribution).

We assume Poisson firing variability, so that $(C^n)_{ii}=\EVb{\var(x_i | s)}=\EVb{\mu_i}=T\EVb{r_i}$, and use a spike-count window $T=100$ms in Fig~\ref{F:diff_tuning} and \ref{F:scatter_C}.

\subsubsection*{Equivalence between penalty functions and constrained optimizations}
In this section we note a standard fact about implementing constrained optimization with penalty functions --- i.e., the method of Lagrange multipliers.

Consider an optimization problem: $\max_{x} f(x)$. Now add a penalty term $p(x)$ with constant $\lambda_0$ and consider the new optimization problem: $\max_{x} f(x) -\lambda_0 p(x)$. If $x_0$ is one of the solutions to this new optimization problem, then it is also an optimal solution to the constraint optimization problem $\max_{x,~p(x)=p(x_0)} f(x)$. 

To show this, let $x_1$ be any point that satisfies $p(x_1)=p(x_0)$. Further, note $x_0$ is also the solution to the problem of $\max_{x} f(x)-\lambda_0(p(x)-p(x_0))$, since we simply add a constant $\lambda_0 p(x_0)$. Therefore,
\beqrn
 f(x_0)-\lambda_0(p(x_0)-p(x_0)) &\ge&  f(x_1)-\lambda_0(p(x_1)-p(x_0)),\\
 f(x_0) &\ge& f(x_1).
\eeqrn

As $x_0$ also satisfies the constraint, we conclude that $x_0$ is an optimal solution to the constrained optimization problem. 

We use this fact to find the information-maximizing noise correlations, with the restriction that the noise correlations by small in magnitude. For a given $\lambda_0$, we perform the optimization $\max_{x} f(x) -\lambda_0 p(x)$, where $f(\cdot)$ in this case is one of our information measures, $x$ refers to the off-diagonal elements of the covariance matrix, and $p(x)$ is the measure of the correlation strength as in Eq.~\eqref{E:p(x)}. Thanks to the above result, we can be assured that the resulting covariance matrix (described by $x$) will be the one that maximizes the information for a particular strength of correlations. By varying $\lambda_0$ (or $r$ in Eq.~\eqref{E:p(x)}), we can thus parametrically explore how the optimal correlation structures change as one allows either larger, or smaller, correlations in the system.

\subsubsection*{Penalty function}
In Section~\namerefquote{S:heter_eg}, our aim is to plot optimized noise correlations at various levels of the correlation strength, as quantified by the Euclidean norm. This constrained optimization problem can be achieved, as shown in the previous section, by adding a term to the information that penalizes the Euclidean norm --- that is, a constant times the sum-of-squares of correlations.  This is precisely the procedure that we follow, ranging over a number of different values of the constant to produce the plot of Fig.~\ref{F:diff_tuning}.  

In more detail, we choose these different values of the constant as follows.  To force the correlations towards a fixed strength of $r$, we optimize a modified objective function with an additional term:
\beq
\label{E:p(x)}
I(C^n) - \frac{\tnorm{ G*V^C }}{2 r} \sum_{i<j} \left(\frac{C^n_{ij}}{V^C_{ij}} \right)^2.
\eeq
As will become clear, the term before the sum is a constant with respect to the terms being optimized; from one optimization to the next, we adjust the value of $r$ in this term. Here the variance terms $V^C_{ij}=\sqrt{C^n_{ii} C^n_{jj}}$ are constants to scale $C^n_{ij}$ properly as correlation coefficients. Also, $G$ is the gradient vector of $I(\cdot)$ at $C^n=D^n$ (the diagonal matrix corresponding independent noise) with respect to off-diagonal entries of $C^n$ (see the remarks after the proof of Theorem \ref{TH:p_or_n}). $G*V^C$ means the entry-wise product of the two vectors (of length $\frac{N(N-1)}{2}$ indexed by $(i,j) ,i< j$). Note that $\tnorm{ \cdot }$ is the ordinary vector 2-norm. 

To understand this choice of the constant in \eqref{E:p(x)}, note that the new optimal correlations with the penalty can be characterized by setting the gradient of the total objective function to 0. In a small neighborhood of $D^n$, the gradient of $I(\cdot)$ is close to $G$. With these substitutions, the equation for the gradient of the total objective function yields approximately:
\beq
\label{E:grad_line}
G-\frac{\tnorm{ G*V^C }}{r} \left\{ \frac{C^n_{ij}}{(V^C_{ij})^2} \right\}_{i<j} =0 \; ,
\; \; \; \; \text{or,}  \; \; \; \; \frac{G*V^C}{\tnorm{ G*V^C }}=\frac{\{\frac{C^n_{ij}}{V^C_{ij}} \}}{ r}.
\eeq
where we took an entry-wise product with $V^C$ and rearranged terms to obtain the final equality.  The final equality implies that the (vector) 2-norm of noise correlations $\{\frac{C^n_{ij}}{V^C_{ij}} \}$ (i.e., the Euclidean norm) is approximately $r$. This is what we set out to achieve with the additional term in the objective function.

\subsubsection*{Rescaling signal correlation}
In Fig.~\ref{F:scatter_C}{\bf DEF}, we make scatter plots comparing noise correlations with the rescaled signal correlations. Here, we explain how and why this rescaling was done. 

First, we note that the rescaling is done by multiplying each signal correlation by a positive weight.  This will not change its sign, the property associated with the sign rule (Fig.~\ref{F:scatter_C}{\bf ABC}).

Next, recall that in deriving the sign rule (Eq.~\eqref{E:dir_prime}), we calculated the gradient of the information with respect to noise correlations.  One should expect alignment between this gradient and the optimal correlations when their magnitudes are small. In other words, if we make a scatter plot with dots whose y and x coordinates are entries of the gradient and noise correlation vectors, respectively (so that the number of dots is the length of these vectors), we expect to see that a straight line will pass through all the dots.

We next note that the entries of the gradient vector $G*V^C$ are not exactly the normalized signal correlations (see Eq.~\eqref{E:grad_line}).  Instead, this vector has additional ``weight factors" that differ for each entry (neuron pair), and hence for each dot in the scatter plot. Thus, to reveal a linear relationship between signal and noise correlations in a scatter plot, we must scale each signal correlation with a proper (positive) weight, determined below, so that $\rhosig_{ij}\rightarrow V^C_{ij}V^G_{ij} \rhosig_{ij}=\text{[universal constant]}\cdot (G*V^C)_{ij}$.  We then redo the scatter plots with these new values  on the horizontal axis. As we will see, the weights $V^C_{ij}V^G_{ij}$ (defined below) do not depend on the noise correlations.  

We now determine $V^G_{ij}$. Recall that our goal is to define $V^G_{ij}$ such that, when it is used to rescale signal correlations as above, we will see a linear alignment between signal and noise correlations.  In other words, if we choose $V^G_{ij}$ correctly, we will have $\frac{C^n_{ij}}{V^C_{ij}} \propto G_{ij}V^C_{ij} = \text{constant}\cdot V^C_{ij} V^G_{ij} \rhosig_{ij}$ (for any $i<j$). Comparing the formulae for $G$ (from the remarks after the proof of Theorem \ref{TH:p_or_n}) and $\rhosig_{ij}$ (Eq.~\eqref{E:rhosig_def_OLE}), we see that $V^G_{ij}=\tnorm{A^0_i} \tnorm{ A^0_{j}}$ satisfies this (with constant$=-\frac{1}{2}$).
Here $A^0=(C^\mu+D^n)^{-1}L$.

\section*{Acknowledgments}
This work was inspired by an ongoing collaboration with Fred Rieke and his colleagues on retinal coding, which suggested the importance of ``mapping" the full space of possible signal and noise correlations.  We gratefully acknowledge the ideas and insights of these scientists. We further wish to thank Andrea Barreiro, Fred Rieke, Kresimir Josic and Xaq Pitkow for helpful comments on this manuscript.

\bibliography{papers}

%

\end{document}